\documentclass[a4paper,twoside,10pt]{article}
\usepackage[a4paper,width=150mm,top=25mm,bottom=25mm,bindingoffset=6mm]{geometry}

\usepackage[english]{babel}

\usepackage{graphicx}
\usepackage[font=small,labelfont=bf]{caption}
\usepackage{subcaption}
\usepackage{color}
\usepackage{mdframed}
\usepackage{framed} 
\usepackage{cancel}
\usepackage[normalem]{ulem}

\usepackage{longtable}
\usepackage{physics}
\usepackage{tabularx}
\usepackage{multirow}
\usepackage{makecell}
\usepackage{booktabs}
\usepackage{setspace}
\newcommand*{\tran}{^{\mkern-1.5mu\mathsf{T}}}

\usepackage{setspace}
\usepackage{lineno,hyperref}
\modulolinenumbers[5]

\usepackage{amsmath,amssymb,amsthm,mathtools,mathrsfs,epsfig,amsfonts,wasysym}
\usepackage{MnSymbol}

\usepackage{ragged2e}
\usepackage{csquotes}
\usepackage{tasks}
\usepackage{paralist}
\usepackage[inline]{enumitem}
\usepackage{tikz-cd}
\usepackage{stmaryrd}

\usepackage[titletoc,title]{appendix}

\newcommand{\de}{\mathrm{d}}
\newcommand{\overbar}[1]{\mkern 1.5mu\overline{\mkern-1.5mu#1\mkern-1.5mu}\mkern 1.5mu}

\usepackage{url,hyperref}
\hypersetup{colorlinks=true,linkcolor=blue,filecolor=magenta,urlcolor=cyan}
\usepackage{nicematrix}
\usepackage{tikz}
\usetikzlibrary{fit}

\usepackage[all]{xy}

\newif\ifbackmatter%
\newcommand{\backmatter}{\global\backmattertrue}%
\usepackage{authblk}
\usepackage{orcidlink}

\newtheorem{secremark}{Remark}[section]

\newtheorem{secproposition}{Proposition}[section]

\newtheorem*{example*}{Example}
\newtheorem*{reminder*}{Reminder}
\newtheorem*{remark*}{Remark}
\newtheorem*{note*}{Note}
\usepackage{titlesec}

\begin{document}

\title{A mathematical study of the role of tBregs in breast cancer}


\author[1,2]{Vasiliki Bitsouni\,\orcidlink{0000-0002-0684-0583}\thanks{\texttt{vbitsouni@math.uoa.gr}}}
\author[1,3]{Nikolaos Gialelis\,\orcidlink{0000-0002-6465-7242}\thanks{\texttt{ngialelis@math.uoa.gr}}}
\author[2]{Vasilis Tsilidis\,\orcidlink{0000-0001-5868-4984}\thanks{\texttt{vtsilidis@outlook.com}}}
\affil[1]{Department of Mathematics, National and Kapodistrian University of Athens, Panepistimioupolis, GR-15784 Athens, Greece}
\affil[2]{School of Science and Technology, Hellenic Open University, 18 Parodos Aristotelous Str., GR-26335 Patras, Greece}
\affil[3]{School of Medicine, National and Kapodistrian University of Athens, Athens,GR-115275, Greece}


\date{}

\maketitle

\begin{abstract}
A model for the mathematical study of immune response to breast cancer is proposed and studied, both analytically and numerically. It is a simplification of a complex one, recently introduced by two of the present authors. It serves for a compact study of the dynamical role in cancer promotion of a relatively recently described subgroup of regulatory B cells, which are evoked by the tumour. 
\end{abstract}

%%================================%%
%% Sample for structured abstract %%
%%================================%%


%\textbf{Keywords:} Cancer Immunology, Tumour Immunity, breast cancer, B cells, T cells, NK cells, immune regulation, Bregs, regulatory B cells, Tregs,  regulatory T cells, tumour-evoked Bregs, tumour carrying capacity, Mathematical Modelling, Stability Analysis, Bifurcation Analysis, Sotomayor's theorem, transcritical bifurcation, numerical simulation 

\maketitle

\section{Introduction} \label{sec1} 

The immune system, divided into the major components of innate and adaptive immunity, grants an organism the ability to detect invading pathogens and transformed cells (e.g. cancer cells) - namely differentiate between ``self" and ``non-self" - and to eliminate them. The provided type and measure of inflammation fit the circumstances, i.e., the immune system dynamically adjust the induced inflammatory response \cite{segel}. The latter function of the immune system is called \textit{immune regulation}. The suppressing mechanisms for immune regulation act as a double-edged sword: on the one hand, they prevent from autoimmune diseases, but on the other hand, they inhibit immune response against cancer \cite{narendra2013immune}. Indeed, it is currently accepted that an aberrant innate and adaptive immune response contributes to tumorigenesis by selecting aggressive clones, inducing immunosuppression, and stimulating cancer cell proliferation and metastasis \cite{gonzalez2018roles}.

The immune system primarily consists of certain types of white blood cells, called lymphocytes. Innate immunity includes lymphocytes such as natural killer (NK) cells, while adaptive immunity includes T and B lymphocytes. Immune regulation involves a special subset of T cells, called regulatory T cells (Tregs) \cite{chraa}. B cells also play a part in immune regulation in the form of regulatory B cells (Bregs) \cite{rosser2015regulatory}. Even though the research on Cancer Immunology has been mainly concentrated on Tregs and NK cells, recently there has been an upsurge in research on Bregs associated with tumours (see, e.g., \cite{guo} and the references therein). 

In this work, we focus on \textit{breast cancer} and we study the role of a relatively recently described \cite{olkhanud2009breast,olkhanud2011tumor} subgroup of Bregs, the \textit{tumour-evoked Bregs} (\textit{tBregs}), in the immune response to the aforementioned type of cancer. \textit{Breast cancer} is more than a hundred times more likely to affect a woman than a man \cite{fentiman2006male} and it is the most common cause of cancer-related death in women world-wide \cite{bray2018global,ferlay2019estimating}. It is estimated that one in eight women living in the United States will develop breast cancer in her lifetime \cite{carol2014breast}. \textit{tBregs} protect metastasizing breast cancer cells from immune effector cells by inducing immune suppression mediated by Tregs \cite{biragyn2014generation}. In a few words, breast cancer seems to generate tBregs, which in turn increase the proliferation of Tregs, which are responsible for the death of the tumour-lysing NK cells, leading eventually to lung metastasis \cite{olkhanud2011tumor}. 

Our study is based on a \textit{dynamical system of non-linear ordinary differential equations}. The theory of ordinary differential equations and dynamical systems has been deployed by scientists in order to gain new insights about the complex interactions between cancer and the immune system. To this end, several models have been constructed for the study  of different entities concerning cancer-immune interactions such as cancer cells, normal cells, cytolytic cells, regulatory cells, cytokines, cancer cells at different stages as well as various therapies \cite{kuznetsov1994, Nani_Freedman_2000, depillis2003mathematical,villasana2003delay, byrne2004macrophage,khailaie2013mathematical, lopez2014validated,denbreems2016, al2020modeling,senekal2021natural}.

The proposed model is a \textit{simplification} of a complex one recently introduced in \cite{bitsouni2021mathematical}. In general, the complexity of the models appeared in bibliography vary greatly. On the one hand, lower dimensional and with simpler functional responses models are advantageous in that they allow for the straightforward application of mathematical techniques, including those of stability and bifurcation analysis \cite{szymanska2003,bunimovich2007, phan2017role,ghosh2018mathematical}. However, their low dimension inhibits the simultaneous investigation of more complex mechanisms that involve multiple components of the tumour micro-environment. On the other hand, higher dimensional models with complex functional responses are more realistic \cite{dePillis2009,dePillis2013,he2017mathematical, makhlouf2020mathematical}. However, they are much more difficult to analyse and can mainly be studied through numerical methods. Nevertheless, both approaches have led to significant insights into the biological mechanisms governing tumour-immune system interactions \cite{eftimie2011interactions}. 

The present study is organised as follows: In \hyperref[ModelFormulation]{Section \ref*{ModelFormulation}}, we construct the mathematical model governing the interactions between breast cancer cells, NK cells, Tregs and tBregs, for which we prove some basic properties of its solution, necessary for ensuring the model's biological relevance, in \hyperref[AppendixSecPropSol]{Appendix \ref*{AppendixSecPropSol}}. In \hyperref[StabAn]{Section \ref*{StabAn}}, we conduct local stability analysis using the linearisation theorem, as well as the centre manifold theorem for the cases in which our equilibrium points are non-hyperbolic. In  \hyperref[BifAn]{Section \ref*{BifAn}}, we investigate the local bifurcations of our model. In \hyperref[reducedModelNumericalSimulationsSection]{Section \ref*{reducedModelNumericalSimulationsSection}}, we perform numerical simulations in order to confirm our qualitative results and further analyse our model. Finally, in \hyperref[Discussion]{Section \ref*{Discussion}}, we sum up and review our results.

\section{Model Formulation} \label{ModelFormulation}
In this section, we develop a relatively simple mathematical model in an attempt to study the interactions between breast cancer cells, NK  cells, Tregs and tBregs. 

Such a model is a simplification of a complex one, recently introduced in \cite{bitsouni2021mathematical}, where seven types of interacting cells are considered, i.e. breast cancer cells, NK  cells, Tregs and tBregs, CD8$^+$ T cells, non-Treg CD4$^+$ T cells, non-Treg CD4$^+$ T cells and non-tBreg B cells. A schematic representation of the interactions between these seven types of cells is given in \hyperref[Fig1]{Figure \ref*{Fig1}}.
%%%%%%%%%%%
\begin{figure}[t]
 \makebox[\textwidth][c]{\includegraphics[width=1.0\textwidth]{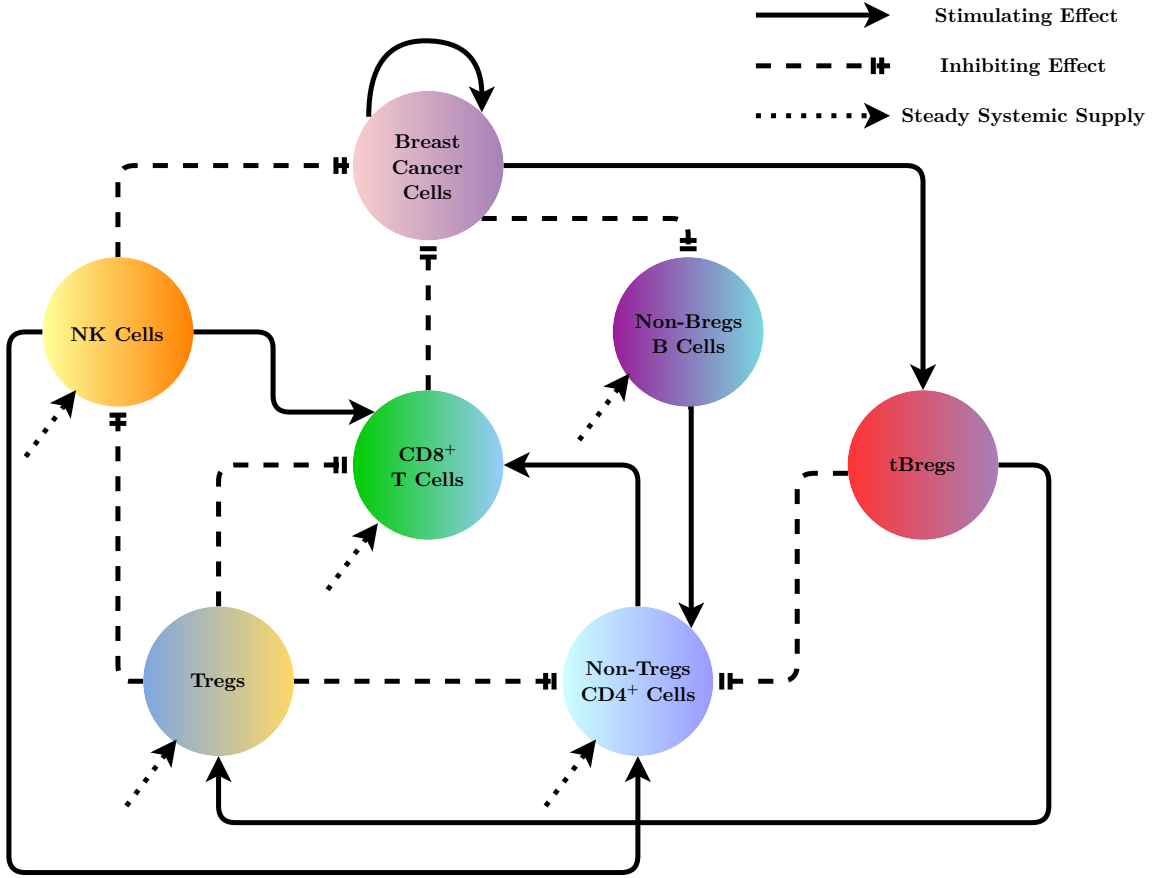}}%
\caption{Interactions between the cells in the model described in \cite{bitsouni2021mathematical}. Solid line (---): Stimulating effect. Dashed line (- -): Inhibiting effect. Dotted line ($\cdot \cdot \cdot$): Steady systemic supply. Figure adapted from \cite{bitsouni2021mathematical}, with the inclusion of the present, yet previously non-depicted, steady systemic supply.}
\label{Fig1}
\end{figure}
%%%%%%%%%%%%%%%
As a first step towards a simplified model of \cite{bitsouni2021mathematical}, we follow \cite{olkhanud2011tumor} and we consider only the interactions between breast cancer cells, NK  cells, Tregs and tBregs, which are now depicted in \hyperref[Fig2]{Figure \ref*{Fig2}}, omitting the rest of the interactions.
%%%%%%%%%%
\begin{figure}[t]
 \makebox[\textwidth][c]{\includegraphics[width=1.0\textwidth]{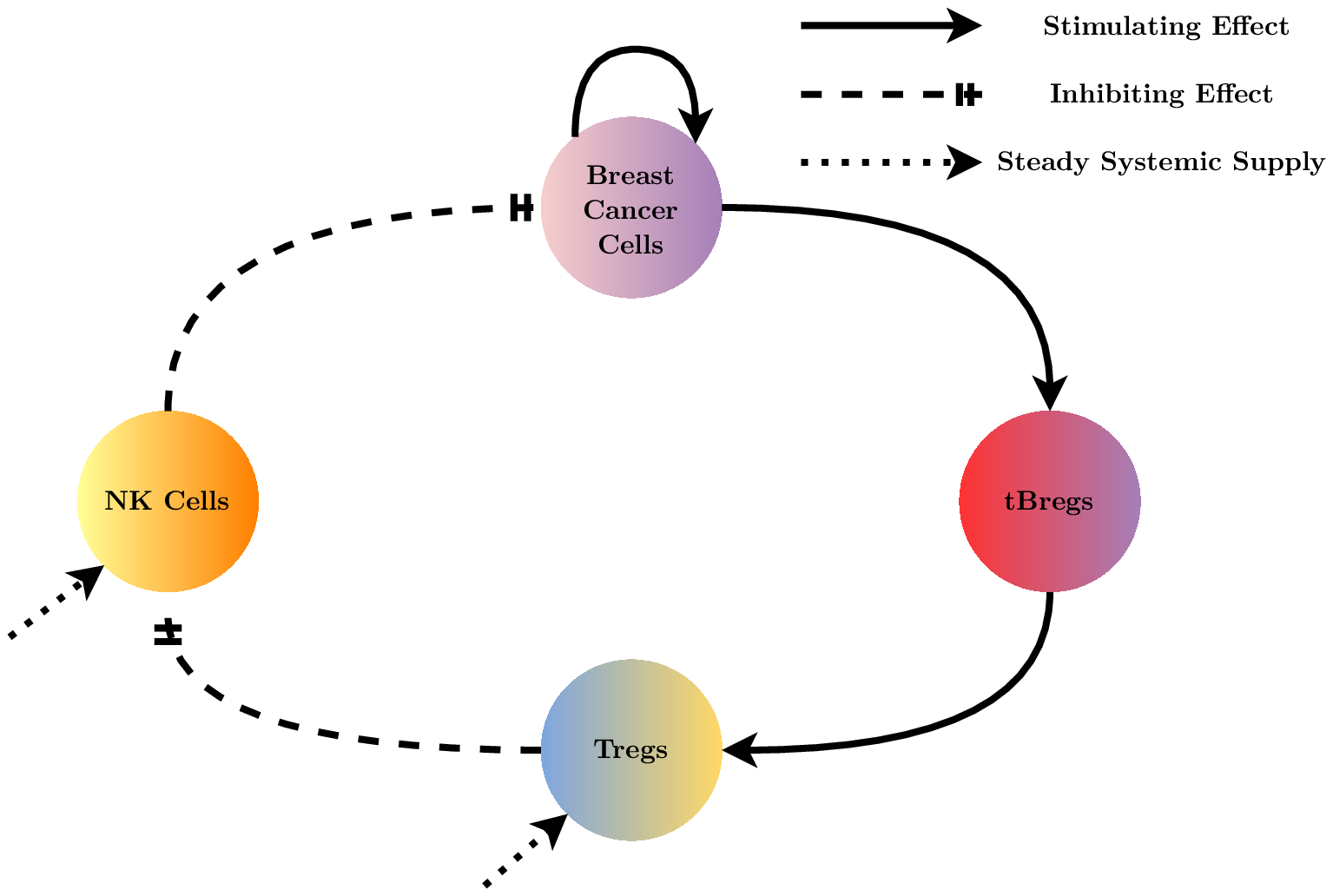}}%
\caption{Interactions between the cells in the simplified model. Solid line (---): Stimulating effect. Dashed line (- -): Inhibiting effect.}
\label{Fig2}
\end{figure}
%%%%%%%%%%%%
Hence, denoting the populations of those cells as $T$, $N$, $R$ and $B$, respectively, and considering them as functions of time, $t$, measured in days, we deduce the following general system of coupled non-linear ordinary equations: 
\begin{align*}  
    & \dv{T}{t} = f_T{\left(T\right)} + f_{T,N}{\left(T,N\right)}\,,      \\ 
    & \dv{N}{t} = f_N{\left(N\right)} + f_{N,R}{\left(N,R\right)} \,, \\
    & \dv{R}{t} = f_R{\left(R\right)} + f_{R,B}{\left(R,B\right)} \,,     \\
    & \dv{B}{t} = f_B{\left(B\right)} + f_{T,B}{\left(T,B\right)} \,,            
\end{align*}
where $f_T$, $f_N$, $f_R$ and $f_B$ stand for the rates of the corresponding cell populations in the absence of interactions, and on the other hand, $f_{T,N}$, $f_{N,R}$, $f_{R,B}$ and $f_{T,B}$ stand for the additional rates in the presence of interactions. 

We then make the following assumptions: 
\begin{enumerate}
\item We utilize the non-interaction terms of \cite{bitsouni2021mathematical}, i.e. we consider logistic growth for the breast cancer cell population, while exponential decay of the rest populations. 
\item As far as the interaction terms are concerned: 
\begin{enumerate}[label=\roman*.]
\item The Treg-induced NK cell inhibition coefficient parameter is eliminated. Besides, to the  best of authors' knowledge, there is no data concerning its values in the bibliography. Moreover, the Hill function, which is used  to model the NK-induced breast cancer cell lysis, is replaced by the simpler linear function. Besides, such a choice is followed in various other models to capture the same dynamics (see \cite{bitsouni2021mathematical} and references therein). 
\item We drop the interaction between breast cancer cells and NK cells, i.e. the NK death mechanism by exhaustion of tumor-killing resources. Besides, in \cite{bitsouni2021mathematical} the value of the corresponding parameter is estimated to be small. Moreover, for the shake of simplicity, we consider linear function  for the modelling of interaction between NK calls and Tregs.
\item As for the interaction term between Tregs and  tBregs, in \cite{bitsouni2021mathematical} a linear function  was used, which involved an intermediate cell type - in particular, non-Treg CD4$^+$ T cells - that is omitted here.  We consider the same functional response without the presence of any intermediate cell population.
\item The condition for the interaction between breast cancer cells and tBregs in \cite{bitsouni2021mathematical} remains unchanged here.
\end{enumerate}
At the end of the day, the rate of every interaction between two cell populations is the product of the size of the two populations, i.e. we consider Holling's type I functional response for every interacting population. 
\end{enumerate} 

To sum up, we get the following system of ordinary differential equations:

\begin{subequations} \label{reducedModel}
\begin{align}  
    & \dv{T}{t} = aT(1-bT) - cNT\,,                    \label{dT}     \\ 
    & \dv{N}{t} = \sigma - \theta_NN - \gamma RN \,,    \label{dN}      \\
    & \dv{R}{t} = \kappa -\theta_RR + m_BBR \,,         \label{dR}      \\
    & \dv{B}{t} = -\theta_BB + m_TTB \,,               \label{dB} 
\end{align}
\end{subequations} 
for some positive constants $a$, $b$, $c$, $\sigma$, $\theta_N$, $\gamma$, $\kappa$, $\theta_R$, $m_B$, $\theta_B$ and $m_T$, along with the initial condition:
\begin{equation}
\left(T\left(0\right),N\left(0\right),R\left(0\right),B\left(0\right)\right)=\left(T_0,N_0,R_0,B_0\right)\in \mathbb{R}_{\ge 0}^4.
    \label{reducedModel_ICs}
\end{equation}

\hyperref[reducedModelParameterTable1]{Table \ref*{reducedModelParameterTable1}} lists all the parameters of our model, along with a brief description about each one, as well as their units. We subsequently give an explanation about each term of our model.

{\footnotesize
\begin{table}[h]
\caption{Description and units of the parameters of system \eqref{reducedModel}.}
\label{reducedModelParameterTable1}
\noindent\makebox[\textwidth]{
\begin{tabular}{p{1cm}p{6cm}p{2cm}}
\toprule
\textbf{Prm.} & \textbf{Description}   & \textbf{Unit}   \\ 
\midrule
$a$  & Tumour growth rate       & day$^{-1}$   \\
$b$      & Inverse of the tumour carrying capacity    & cell$^{-1}$  \\
$c$   & Rate of NK-induced tumour death  & cell$^{-1} \cdot$ day$^{-1}$\\ 
\midrule
$\sigma$           & Constant source of NK cells                    & cell $\cdot$ day$^{-1}$          \\
$\theta_N$         & Rate of programmable NK cell death             & day$^{-1}$                        \\
$\gamma$           & Rate of NK cell death due to Tregs             & cell$^{-1} \cdot$ day$^{-1}$ \\ 
\midrule
$\kappa$           & Constant source of Tregs                       & cell $\cdot$ day$^{-1}$         \\
$\theta_R$         & Rate of programmable Treg death                & day$^{-1}$           \\
$m_B$              & Rate of tBreg-induced Treg activation          & cell$^{-1} \cdot$ day$^{-1}$   \\
\midrule
$\theta_B$         & Rate of programmable tBreg cell death          & day$^{-1}$                       \\
$m_T$              & Rate of breast-cancer-induced tBreg activation & cell$^{-1} \cdot$ day$^{-1}$ \\
\bottomrule
\end{tabular}}
\end{table}
}

The first term of the right-hand side of equation \eqref{dT} models the logistic growth of breast cancer cells, while the second term models the death of breast cancer cells due to NK cells. We can show that (see \hyperref[BofT]{Proposition \ref*{BofT}})  the set $\left[0,1/b\right]$ is positively invariant for the component $T$ of the solution of  initial value problem  $\left\{\eqref{reducedModel},\eqref{reducedModel_ICs}\right\}$ (from now on: IVP). Hence, by analogy with the terminology used for the logistic equation, we say that the (positive) constant $1/b$ is the \textit{tumour carrying capacity} for IVP. It expresses a maximal reachable size due to competition between adjacent cells, e.g. for space or nutrients \cite{vaghi2020carrying}.  In general, the metastatic potential of a tumour increases as it reaches its carrying capacity. The conventional linear correlation between primary tumour size and the likelihood of metastasis (to the lymph nodes or to distant sites) is questioned, as a characteristic non-linear relationship, with a sigmoid corresponding curve, has already been described (see, e.g., \cite{sopik2018relationship} and the references therein).

As far as equation \eqref{dN} is concerned, parameter $\sigma$ models the steady source of NK cells. The second term models the natural death of NK cells, whereas the last term models the Treg-induced NK apoptosis.

Regarding equation \eqref{dR}, much like equation \eqref{dN}, parameter $\kappa$ models the steady source of Tregs. The second term models the natural death of Tregs, whereas the last term models the proliferation of Tregs due to tBregs.

Finally, the first term of equation \eqref{dB} models the natural death of tBregs, whereas the last term models their proliferation due to the existence of breast cancer.

\section{Equilibria and Stability Analysis} \label{StabAn}
We continue our analysis by finding the equilibrium points of system \eqref{reducedModel}, which we write as $E=\left(\overbar{T}, \overbar{N}, \overbar{R}, \overbar{B}\right)$. To do so, we equate the right-hand side of \eqref{reducedModel} to $\mathbf{0}$ and solve the resulting system of algebraic equations

\begin{subequations}
\begin{align}  
     a\overbar{T}(1-b\overbar{T}) - c\overbar{N}\overbar{T}             & = 0    \,,   \label{Tss}  \\ 
     \sigma - \theta_N\overbar{N} - \gamma \overbar{R}\overbar{N} & = 0 \,,  \label{Nss} \\
     \kappa -\theta_R\overbar{R} + m_B\overbar{B}\overbar{R}     & = 0\,,    \label{Rss}  \\
     -\theta_B\overbar{B} + m_T\overbar{T}\overbar{B}           &  = 0 \,.\label{Bss}
\end{align}
\end{subequations} 

\begin{secproposition}
System \eqref{reducedModel} has three equilibrium points

\begin{align*}
E_1 & =   \left(0, \frac{a}{c_1}, \frac{\kappa }{\theta _R}, 0\right) \, , \\
E_2 & = \left(\frac{c_1-c}{bc_1}, \frac{a}{c_1}, \frac{\kappa }{\theta _R}, 0\right) \text{\quad and} \\
E_3 & =
\begin{multlined}[t]
\biggl(\frac{\theta _B}{m_T}, -\frac{a   \Theta_1}{c m_T}, -\frac{ \Theta_2}{a \gamma   \Theta_1},  \frac{\sigma m_T\theta_R\left(c-c_2\right)}{m_B\Theta_2} \bigg) \, ,
\end{multlined}
\end{align*}
where 
\begin{align*}
   \Theta_1 & \coloneqq  b \theta_B - m_T  \,, \\
   \Theta_2 &  \coloneqq  a \theta_N \Theta_1 + c \sigma  m_T\,, \\
    c_1 & \coloneqq  \frac{a\left(\theta_N\theta_R+\gamma\kappa\right)}{\theta_R \sigma} \,, \\
    c_2 & \coloneqq -\frac{\Theta_1 c_1}{m_T}\,,
\end{align*}
for which we can easily see that $c_1 > c_2=\left(1-\frac{b\theta_B}{m_T}\right)c_1$. 
\end{secproposition}

\begin{proof}
 From equation \eqref{Tss}, we see that either $\overbar{T} = 0$ or $ a(1-b\overbar{T}) - c\overbar{N} = 0 $. In the case where $\overbar{T} = 0$, by replacing $\overbar{T} = 0$ to equation \eqref{Bss}, we get that $\overbar{B} = 0$, which we consequently replace to equation \eqref{Rss} to get $\overbar{R} = \kappa/\theta_R$. Replacing the derived value of $\overbar{R}$ to equation \eqref{Nss}, we get that $\overbar{N} = (\sigma  \theta _R)/(\gamma  \kappa +\theta _N \theta _R)$, which yields the equilibrium point $E_1$.
 
 We next assume that $\overbar{T} \neq 0$, thus $ a(1-b\overbar{T}) - c\overbar{N} = 0 $. From equation \eqref{Bss}, we see that either $ \overbar{B} = 0 $ or $ \overbar{T} = \theta_B / m_T $. In the case where $\overbar{B} = 0$, we follow the same procedure as in the above paragraph to find the values of $\overbar{N}$ and $\overbar{R}$, which are equal to their respective counterpart from equilibrium $E_1$. Then, by replacing $\overbar{N}$ to equation $ a(1-b\overbar{T}) - c\overbar{N} = 0 $, we get $\overbar{T} = (a \gamma  \kappa +a \theta _N \theta _R-c \sigma  \theta _R)/(a b \gamma  \kappa +a b \theta _N \theta _R)$ which gives us the second equilibrium point, $E_2$.
 
 Finally, we have the case where $\overbar{T} \neq 0$ and $\overbar{B} \neq 0$. From equation \eqref{Bss}, we see that $\overbar{T} = \theta_B / m_T$. By replacing it to equation $ a(1-b\overbar{T}) - c\overbar{N} = 0 $, we find $ \overbar{N} = \left(a \left(m_T-b \theta _B\right)\right)/\left(c m_T\right) $. Subsequently, replacing $\overbar{N}$ to \eqref{Nss}, we get that $\overbar{R} = (a b \theta _B \theta _N-a \theta _N m_T+c \sigma  m_T)/(a \gamma  \left(m_T-b \theta _B\right))$ and by replacing $\overbar{R}$ to equation \eqref{Rss}, we find $\overbar{B} = (a \left(m_T-b \theta _B\right) \left(\gamma  \kappa +\theta _N \theta _R\right)-c \sigma  m_T \theta _R)/(m_B \left(a \theta _N \left(m_T-b \theta _B\right)-c \sigma  m_T\right))$, which gives us the third equilibrium point, $E_3$.
\end{proof}

Since we are only interested in the cases were our variables are non-negative, we determine the conditions under which the equilibrium points' coordinates are non-negative.

\begin{secproposition}
Admissible (i.e. with non-negative components) equilibrium point $E_1$ exists for any set of parameters.
\end{secproposition}

\begin{proof}
 It is evident that the coordinates of $E_1$ are non-negative since all the parameters are positive.
\end{proof}

\begin{secproposition} \label{existE2}
Admissible equilibrium point $E_2$ exists when $c\leq c_1$. If $ c=c_1$, then $E_1 \equiv E_2$ and if $  c<c_1$, then $E_1 \nequiv E_2$.
\end{secproposition}

\begin{proof}
 In order for $E_2$ to have non-negative coordinates, the numerator of its first coordinate needs to be non-negative, thus $ c_1-c\geq 0$. However, if $ c_1-c = 0$, then the first coordinate of $E_2$ becomes zero, therefore $E_1 \equiv E_2$. Hence, $E_2$ exists and does not coincide with $E_1$ when $ c<c_1$.
\end{proof}

\begin{secproposition} \label{existE3}
Admissible equilibrium point $E_3$ exists when $c\geq c_2$ and $\frac{\theta_B}{m_T} < \frac{1}{b}$. If $c=c_2$, then $E_2 \equiv E_3$ and if $c>c_2$, then $E_2 \nequiv E_3$. 
\end{secproposition}
 
\begin{proof}
 Clearly, the first coordinate of $E_3$ is positive.
 For its second coordinate to be non-negative, $\Theta_1\leq  0$ must hold. However, when  $\Theta_1 =  0$, the denominator of the third coordinate of $E_3$ becomes zero. Therefore, assuming that $\Theta_1\neq  0$, we have that $\Theta_1 <  0$ or equivalently $ \frac{\theta_B}{m_T} < \frac{1}{b}$.
 
 Since $\Theta_1 <  0$, then the numerator of the third coordinate of $E_3$ needs to be non-negative, thus $\Theta_2 \geq 0$. However, if $\Theta_2= 0$, then the denominator of the fourth coordinate of $E_3$ becomes zero. Therefore, assuming that $\Theta_2 \neq 0$, we have that $\Theta_2 > 0$ or equivalently $\frac{a \theta_N - c \sigma}{a b \theta_N} < \frac{\theta_B}{m_T}$.
 
 Since we assumed that $\Theta_2>0$, for its  forth coordinate to be non-negative, $c\geq c_2$ must hold or equivalently $\frac{a \gamma  \kappa +a \theta _N \theta _R-c \sigma  \theta _R}{a b \gamma  \kappa +a b \theta _N \theta _R} \le \frac{\theta_B}{m_T}$.
 
 We need to determine which of the two lower bounds of $\frac{\theta_B}{m_T}$ we found is the greatest, so we can take that as our lower bound for $\frac{\theta_B}{m_T}$. Let $\frac{a \gamma  \kappa +a \theta _N \theta _R-c \sigma  \theta _R}{a b \gamma  \kappa +a b \theta _N \theta _R} < \frac{a \theta_N - c \sigma}{a b \theta_N}$. It follows that
 
\begin{align*}
 \frac{a \gamma  \kappa +a \theta _N \theta _R-c \sigma  \theta _R}{\gamma  \kappa + \theta _N \theta _R} < \frac{a \theta_N - c \sigma}{ \theta_N} 
    &\Leftrightarrow   a-\frac{ c \sigma  \theta _R}{\gamma  \kappa + \theta _N \theta _R}<a-\frac{ c \sigma}{ \theta_N}   \Leftrightarrow \gamma  \kappa < 0
\end{align*}
 which is impossible, therefore we have that $\frac{a \gamma  \kappa +a \theta _N \theta _R-c \sigma  \theta _R}{a b \gamma  \kappa +a b \theta _N \theta _R} \ge \frac{a \theta_N - c \sigma}{a b \theta_N}$ and we choose $\frac{a \gamma  \kappa +a \theta _N \theta _R-c \sigma  \theta _R}{a b \gamma  \kappa +a b \theta _N \theta _R}$ as our lower bound for $\frac{\theta_B}{m_T}$.
 
However, when $\frac{a \gamma  \kappa +a \theta _N \theta _R-c \sigma  \theta _R}{a b \gamma  \kappa +a b \theta _N \theta _R} = \frac{\theta_B}{m_T}$, we have that $E_2 \equiv E_3$. Therefore, when $E_2$ does not coincide with $E_3$, then $\frac{a \gamma  \kappa +a \theta _N \theta _R-c \sigma  \theta _R}{a b \gamma  \kappa +a b \theta _N \theta _R} < \frac{\theta_B}{m_T}$, and the proposition is proved.
\end{proof}

We next study the local stability of system \eqref{reducedModel} using the linearisation theorem, as well as the center manifold theorem in the cases where the equilibrium point is non-hyperbolic, while also utilising Descartes' rule of signs in the case where the Jacobian matrix's eigenvalues are too complex to be computed.

We begin by computing the Jacobian matrix of system \eqref{reducedModel} to be equal to
\begin{equation*}
\mathbf{J}(T,N,R,T) =
\begin{bmatrix} 
 a-2 a b T-c N  & -c T & 0 & 0 \\
 0 &  -\theta_N-\gamma R  & - \gamma N   & 0 \\
 0 & 0 & -\theta _R+m_B B  &  m_B R \\
 m_T B  & 0 & 0 & -\theta_B+m_T T  \\
\end{bmatrix} \, .
\end{equation*}

\begin{secproposition}
Equilibrium point $E_1$, when it does not coincide with $ E_2$, is locally asymptotically stable (stable node) if $c>c_1 $ and unstable (saddle) if $c<c_1$.
\end{secproposition}
\begin{proof}
 At the equilibrium point $E_1$, the Jacobian matrix becomes
\begin{equation} \label{jacE1}
    \mathbf{J}(E_1;c) = 
    \begin{bmatrix}
 \frac{a\left(c_1-c\right)}{c_1} & 0 & 0 & 0 \\
 0 & -\frac{\sigma c_1}{a} & -\frac{a\gamma}{c_1} & 0 \\
 0 & 0 & -\theta _R & \frac{\kappa  m_B}{\theta _R} \\
 0 & 0 & 0 & -\theta _B \\
    \end{bmatrix} \,,
\end{equation}
with corresponding eigenvalues
\begin{equation}\label{eigenvalE1}
   \lambda_{11}\left(c\right) = \frac{a\left(c_1-c\right)}{c_1},\;  \lambda_{12} = -\theta _B\,, \; \lambda_{13} = -\theta _R\quad \text{ and }\quad \lambda_{14} =  -\frac{\sigma c_1}{a}\,.
\end{equation}

 It is clear that the  eigenvalues $\lambda_{12}, \lambda_{13}$ and $\lambda_{14}$ are always negative since all the parameters are positive, while $\lambda_{11}\left(c\right)$ is negative when  $c>c_1 $ and positive when $c<c_1 $.
\end{proof}

\begin{secremark}
Note that the corresponding eigenvectors of matrix \eqref{jacE1} are
 \begin{align} 
      \mathbf{u}_{11}  &= \begin{bmatrix}
     1,0,0,0
     \end{bmatrix}\tran, \quad\mathbf{u}_{12}  = \begin{bmatrix}
      0,-\frac{a^2\gamma  \kappa  m_B}{ c_1\theta _R\left(\theta _B-\theta _R\right)\left(a\theta _B-\sigma c_1\right)},-\frac{\kappa  m_B}{\theta _R \left(\theta _B-\theta _R\right)},1 
     \end{bmatrix}\tran, \label{eigenvecE1.a}\\
     \mathbf{u}_{13} & = \begin{bmatrix}
      0,\frac{a^2\gamma }{c_1\left(a\theta _R-\sigma c_1\right)},1,0
     \end{bmatrix}\tran \quad \text{ and } \quad
     \mathbf{u}_{14} = \begin{bmatrix}
     0,1,0,0
     \end{bmatrix}
    \tran.\label{eigenvecE1.b}
 \end{align}
In the case where $\lambda_{11}$ is positive, we have that the local stable invariant manifold is tangent to the stable manifold of the linearised system, which is $E^s = \text{span}\{\mathbf{u}_{12},\mathbf{u}_{13},\mathbf{u}_{14}\}$, while the local unstable invariant manifold is tangent to the unstable manifold of the linearised system, which is $E^u = \text{span}\{\mathbf{u}_{11}\}$, that is the $T$ axis. Biologically, this means that a small initial number of tumour cells will increase even with the presence of normal levels of NK cells, Tregs and tBregs.
\end{secremark}

\begin{secproposition}
Admissible equilibrium point $E_2$, when it  does not coincide with $E_3$, is locally asymptotically stable (stable node) if  $c>c_2$ and unstable (saddle) if  $c<c_2$.
\end{secproposition}
\begin{proof}
At the equilibrium point $E_2$, the Jacobian matrix becomes
 \begin{equation*}% \label{jacE2}
    \mathbf{J}(E_2;c) = 
\begin{bmatrix}
 -\frac{a\left(c_1-c\right)}{c_1}  &  -\frac{c\left(c_1-c\right)}{bc_1}  & 0 & 0 \\
 0 & -\frac{\sigma c_1}{a} & -\frac{a\gamma}{c_1} & 0 \\
 0 & 0 & -\theta _R & \frac{\kappa  m_B}{\theta _R} \\
 0 & 0 & 0 & \frac{m_T\left(c_2-c\right)}{bc_1} \\
\end{bmatrix} \, ,
\end{equation*}
with corresponding eigenvalues
\begin{equation} 
     \lambda_{21}\left(c\right)=\frac{m_T\left(c_2-c\right)}{bc_1}\,, \; \lambda_{22}=-\theta _R\,, \; \lambda_{23} =-\frac{\sigma c_1}{a} \, \quad \text{ and }  \lambda_{24}\left(c\right)=-\frac{a\left(c_1-c\right)}{c_1}   \,.
     \label{eigenvalE2}
 \end{equation}
 
 It is clear that the  eigenvalues $\lambda_{22}, \lambda_{23}$  are always negative since all the parameters of our model are positive. Moreover, $\lambda_{24}\left(c\right)$ is also negative, since its denominator is positive and its numerator is positive when $E_2$ exists from \hyperref[existE2]{Proposition \ref*{existE2}}. Finally,  $\lambda_{21}\left(c\right)$ is negative when $c>c_2$ and positive when $c<c_2$. 
 
\end{proof}

\begin{secproposition}
Admissible equilibrium point $E_3$, when it  does not coincide with $E_2$, is always unstable.
\end{secproposition}

\begin{proof}
 At the equilibrium point $E_3$, the Jacobian matrix becomes
 
\begin{equation*}
 \mathbf{J}(E_3;c) = 
 \begin{bmatrix}
 -\frac{a b \theta _B}{m_T} & -\frac{c \theta _B}{m_T} & 0 & 0 \\
 0 & \frac{c \sigma  m_T}{a \Theta_1} & \frac{a \gamma \Theta_1}{c m_T} & 0 \\
 0 & 0 & \frac{a \gamma  \kappa  \Theta_1}{\Theta_2} & -\frac{m_B \Theta_2}{a \gamma \Theta_1} \\
\frac{m_T\left(a \gamma  \kappa \Theta_1+\theta_R \Theta_2  \right)}{m_B \Theta_2} & 0 & 0 & 0 \\
\end{bmatrix}\,,
\end{equation*}
with its characteristic polynomial being
\begin{equation} \label{charPol}
    \lambda^4 + \alpha_3 \lambda^3 + \alpha_2 \lambda^2 + \alpha_1 \lambda + \alpha_0 = 0\,,
\end{equation} where
\begin{align*}
    \alpha_0 &= -\frac{ \theta_B\left(a \gamma  \kappa \Theta_1+\theta_R \Theta_2  \right)}{m_T}\,, \\
    \alpha_1 &= \frac{a b c \gamma  \kappa  \sigma  \theta _B }{\Theta_2}\,, \\
    \alpha_2 &= -\frac{a^2 b \gamma  \kappa  \theta _B \Theta_1}{m_T \Theta_2}+\frac{c \gamma  \kappa  \sigma  m_T }{\Theta_2}-\frac{b c \sigma  \theta _B}{\Theta_1}\,, \\
    \alpha_3 &= -\frac{a \gamma  \kappa \Theta_1}{\Theta_2}-\frac{c \sigma  m_T}{a \Theta_1}+\frac{a b \theta _B}{m_T}\,.
\end{align*}
We utilise Descartes' rule of signs to prove our claim. Firstly, let $\alpha_0 \ge 0$, which means that
\begin{equation*} 
   \frac{\theta_B}{m_T} \le \frac{a \gamma  \kappa +a \theta _N \theta _R-c \sigma  \theta _R}{a b \gamma  \kappa +a b \theta _N \theta _R}=\dfrac{c_1-c}{bc_1}\,,
\end{equation*}
which is impossible when $E_3$ exists and does not coincide with $E_2$ from \hyperref[existE3]{Proposition \ref*{existE3}}. Hence,  $\alpha_0 < 0$. Consequently, from   \hyperref[existE3]{Proposition \ref*{existE3}}, we have that when $E_3$ exists and does not coincide with $E_2$, then $\Theta_1 < 0$ and $\Theta_2 > 0$, which means that $\alpha_1,\alpha_2,\alpha_3 > 0$. Next, we apply the transformation
\begin{equation*}
    \lambda \mapsto -\lambda \;,
\end{equation*}
to \eqref{charPol}, to get
\begin{equation} \label{charPolTrans}
   \underbrace{\lambda^4}_{+}  \underbrace{- \alpha_3 \lambda^3}_{-} + \underbrace{\alpha_2 \lambda^2}_{+}  \underbrace{-\alpha_1 \lambda}_{-}  \underbrace{+\alpha_0}_{-} = 0\,.
\end{equation}
The sign changes in the sequence of the polynomial coefficients of \eqref{charPolTrans} are three. Therefore, \eqref{charPolTrans} has exactly one or three positive real roots, which means that \eqref{charPol} has exactly one or three negative real roots. In either case, \eqref{charPol} has at least one positive real root, so it follows that $E_3$ is unstable.
\end{proof}

We have yet to tackle the cases in which $E_1 \equiv E_2$ and $E_2 \equiv E_3$. When $ c = c_1,$ then $ \lambda_{11}\left(c_1\right) = 0$, and when $c=c_2,$ then $ \lambda_{21}\left(c_2\right) = 0$, so in each case a one dimensional center manifold arises, and the linearisation theorem cannot be used. In order to determine the dynamics of our system in those two cases, we utilise the center manifold theorem \cite{jordan2007nonlinear}.

\begin{secproposition}
 When $c = c_1,$ then the equilibrium point $E_1 \equiv E_2$ is locally asymptotically stable.
\end{secproposition}

\begin{proof}
  For convenience, we firstly move the equilibrium point to the origin, using the transformation 
 \begin{equation} \label{transform1}
   \begin{aligned}
T & \mapsto T+\overbar{T}=T \, , \\
N & \mapsto N+\overbar{N}=N+\frac{a}{c_1} \, , \\
R & \mapsto R+\overbar{R}=R+\frac{\kappa }{\theta _R} \, , \\
B & \mapsto B+\overbar{B}=B \, .
\end{aligned}   
\end{equation}
 
 After replacing \eqref{transform1} to system \eqref{reducedModel}, while utilising that $c=c_1$, we get the transformed system
 
\begin{equation} \label{trSys1}
    \dv{\mathbf{X}}{t} = \mathbf{A_1 X} + \mathbf{P_1}\,,
\end{equation}
where
\begin{equation*}
    \mathbf{X} = \begin{bmatrix}
     T,N,R,B
     \end{bmatrix}\tran\,,
\end{equation*}

\begin{equation*}% \label{jacobCM1}
    \mathbf{A_1} \coloneqq  \mathbf{J}(E_1;c=c_1)=
\begin{bmatrix}
 0 & 0 & 0 & 0 \\
 0 & -\frac{\sigma c_1}{a} & -\frac{a\gamma}{c_1} & 0 \\
 0 & 0 & -\theta _R & \frac{\kappa  m_B}{\theta _R} \\
 0 & 0 & 0 & -\theta _B \\
    \end{bmatrix} 
\end{equation*}
and 
\begin{equation*}
\mathbf{P_1} = \begin{bmatrix}
 -abT^2-c_1TN, -\gamma RN, m_B  BR ,  m_T  TB 
\end{bmatrix}\tran\,.
\end{equation*}

$\mathbf{A_1}$ is upper diagonal and its eigenvalues  are $\lambda_{11}\left(c_1\right), \lambda_{12}, \lambda_{13}$ and $\lambda_{14}$,  with corresponding eigenvectors $\mathbf{u}_{11}, \mathbf{u}_{12}, \mathbf{u}_{13}$ and $\mathbf{u}_{14}$, as given by \eqref{eigenvalE1}-\eqref{eigenvecE1.b}.

 We bring system \eqref{trSys1} to its normal form with the help of the transformation
 \begin{equation} \label{transNormal1}
     \mathbf{X} = \mathbf{U_1} \mathbf{Y}\,,
 \end{equation}
 where we let 
  \begin{equation*}
     \mathbf{Y} = \begin{bmatrix}
      y_1, y_2, y_3, y_4
     \end{bmatrix}\tran
      \quad \text{ and } \quad
      \mathbf{U_1} =  \begin{bmatrix}
      \mathbf{u}_{11}, \mathbf{u}_{12}, \mathbf{u}_{13}, \mathbf{u}_{14}
     \end{bmatrix}\,.
 \end{equation*}
 By replacing \eqref{transNormal1} to system \eqref{trSys1}, we have that
  \begin{alignat}{2}
                    && \mathbf{U_1}\dv{\mathbf{Y}}{t} &= \mathbf{A_1 U_1 Y} + \mathbf{P_1} \nonumber\\
    \Leftrightarrow && \dv{\mathbf{Y}}{t}   &= \mathbf{U_1^{-1} A_1 U_1 Y} + \mathbf{U_1^{-1} P_1}   \nonumber\\
    \Leftrightarrow && \dv{\mathbf{Y}}{t}   &= \mathbf{\Lambda_1 Y} + \mathbf{F_1} \,, \label{NormalForm}
 \end{alignat}
 where 
 \begin{align*}
    \mathbf{\Lambda_1} &= \text{diag}(\lambda_{11}\left(c_1\right),\lambda_{12},\lambda_{13},\lambda_{14}) \quad \text{ and } \\
    \mathbf{F_1} &= \begin{bmatrix}
      f_{11}, f_{12}, f_{13}, f_{14}
     \end{bmatrix}\tran = \mathbf{U_1^{-1} P}\,
 \end{align*}

 $\mathbf{F_1}$ is a function of $T, N, R$ and  $B$. In order to express $\mathbf{F_1}$ as a function of $ y_1, y_2, y_3$ and $y_4$, we substitute $T, N, R$ and  $B$ from equation \eqref{transNormal1} to $\mathbf{F_1}$, to find that 
 
 \begin{equation*}% \label{F1Matrix}
 \mathbf{F_1} = 
\begin{bmatrix}
    -a b y_1^2+ w_{11} y_1 y_2 + w_{12} y_1 y_3- c_1 y_1 y_4  \\
    m_T y_1 y_2  \\
    w_{13} y_1 y_2 + w_{14} y_2^2 + m_B y_2 y_3\\
    w_{15} y_1 y_2 + w_{16} y_2 y_3 + w_{17} y_2^2 + w_{18} y_2 y_4 + w_{19} y_3^2 + w_{110} y_3 y_4
\end{bmatrix}\,,
 \end{equation*}
 where $w_{1i}, \;  i=1,2,3,...,10$ are known constants.
 
 Hence, system \eqref{NormalForm} can be written in the form
 \begin{equation*}
     \dv{\mathbf{Y}}{t} = \begin{bmatrix}
    \mathbf{B_1} & \mathbf{0} \\
   \mathbf{0} & \mathbf{C_1}   \\
\end{bmatrix} \mathbf{Y} + \mathbf{F_1}\,,
 \end{equation*}
 or equivalently 
 \begin{subequations}
\begin{align}
      \dv{y_1}{t} & = \mathbf{B_1} y_1 + f_{11}   \,,    \label{y1'} \\   
    \begin{bmatrix}
        \dv{y_2}{t} \\ \dv{y_3}{t} \\ \dv{y_4}{t}
    \end{bmatrix}  
    & =
   \mathbf{C_1} \begin{bmatrix}
        y_2 \\ y_3 \\ y_4
    \end{bmatrix} 
     + \begin{bmatrix}
        f_{12} \\ f_{13} \\ f_{14}
    \end{bmatrix}\,,
  \end{align}
 \end{subequations}
 where
 \begin{equation*}
     \mathbf{B_1} = \lambda_{11}\left(c_1\right) = 0 \quad \text{ and } \quad \mathbf{C_1} = \text{diag}(\lambda_{12},\lambda_{13},\lambda_{14})\,.
 \end{equation*}
 
 We have that $\mathbf{B_1} $ and $\mathbf{C_1}$ are constant matrices, with the eigenvalues of $\mathbf{B_1} $ having zero real part and the eigenvalues of $\mathbf{C_1}$ having negative real part, whereas $f_{1i}$ is smooth with $f_{1i}(0,0,0,0) = 0$ and $Df_{1i} (0,0,0,0) = 0$ for $i=1,2,3,4$.
 From the center manifold theorem, there exists a center manifold which is parametrised by 
 \begin{equation*}
     \begin{bmatrix}
        y_2 \\ y_3 \\ y_4
    \end{bmatrix}  = \mathbf{h_1}(y_1) = \begin{bmatrix}
        h_{12}(y_1) \\ h_{13}(y_1) \\ h_{14}(y_1)
    \end{bmatrix}\,,
 \end{equation*}
 with $\mathbf{h_1}(0)=0$ and $D\mathbf{h_1}(0) = 0$, and satisfying 
 \begin{equation*}
     \mathbf{C_1} \cdot \mathbf{h_1}(y_1) +  \begin{bmatrix}
        f_{12}(y_1, \mathbf{h_1}(y_1) ) \\ f_{13}(y_1, \mathbf{h_1}(y_1) ) \\ f_{14}(y_1, \mathbf{h_1}(y_1) ) 
    \end{bmatrix} = D\mathbf{h_1}(y_1) \cdot \left[\mathbf{B_1} y_1 + f_{11}(y_1, \mathbf{h_1}(y_1)) \right]\,,
 \end{equation*}
while the flow on the center manifold is defined by the differential equation \eqref{y1'}.

Since $\mathbf{h_1}(0)=0$ and $D\mathbf{h_1}(0) = 0$, then by approximating the center manifold with a Taylor series around 0, we get
\begin{equation} \label{taylorApprox}
    \mathbf{h_1}(y_1) = \begin{bmatrix}
        b_{12} y_1^2 + b_{13} y_1^3 + b_{14} y_1^4 + O(y_1^5) \\ c_{12} y_1^2 + c_{13} y_1^3 + c_{14} y_1^4 + O(y_1^5) \\ d_{12} y_1^2 + d_{13} y_1^3 + d_{14} y_1^4 + O(y_1^5)
    \end{bmatrix}\,.
\end{equation}
By replacing \eqref{taylorApprox} to \eqref{y1'} in order to express the terms of $y_2, y_3, y_4$ in $f_{11}$, we get

\begin{equation} \label{flowT1}
    \dv{y_1}{t} = -aby_1^2 + O(y_1^3)\,.
\end{equation}

From  $\mathbf{Y} = \mathbf{U_1^{-1} X}$, we see that $y_1 = T$, so $y_1$ is non-negative. That means that 0 is an asymptotically stable point of equation \eqref{flowT1}, when approached from $y_1 \ge 0$ since $-ab<0$. Thus, from the center manifold theorem, $\mathbf{0}$ is a locally asymptotically stable equilibrium of system \eqref{trSys1}, which proves the proposition.

\end{proof}

\begin{secproposition}
  When $c=c_2$, then the equilibrium point $E_2 \equiv E_3$ is unstable.
\end{secproposition}

\begin{proof}
For convenience, we firstly move the equilibrium point to the origin, using the transformation 
\begin{equation} \label{transform2}
   \begin{aligned}
T & \mapsto T+\overbar{T}=T+\frac{c_1-c}{bc_1}  \, , \\
N & \mapsto N+\overbar{N}=N+\frac{a}{c_1} \, , \\
R & \mapsto R+\overbar{R}=R+\frac{\kappa}{\theta_R} \, , \\
B & \mapsto B+\overbar{B}=B \, .
\end{aligned}   
 \end{equation}

 After replacing \eqref{transform2} to system \eqref{reducedModel}, while utilising that $c_2=c$, we get the transformed system

\begin{equation} \label{trSys2}
    \dv{\mathbf{X}}{t} = \mathbf{A_2 X} + \mathbf{P_2}\,,
\end{equation}
where
\begin{equation*}
    \mathbf{X} = \begin{bmatrix}
     T,N,R,B
     \end{bmatrix}\tran\,,
\end{equation*}

\begin{equation*}
    \mathbf{A_2} \coloneqq \mathbf{J}(E_2;c=c_2)= 
\begin{bmatrix}
 -\frac{a\left(c_1-c_2\right)}{c_1}  &  -\frac{c_2\left(c_1-c_2\right)}{bc_1}  & 0 & 0 \\
 0 & -\frac{\sigma c_1}{a} & -\frac{a\gamma}{c_1} & 0 \\
 0 & 0 & -\theta _R & \frac{\kappa  m_B}{\theta _R} \\
 0 & 0 & 0 & 0
\end{bmatrix}
\end{equation*}
and 
\begin{equation*}
\mathbf{P_2} = \begin{bmatrix}
 -abT^2-c_2TN, -\gamma RN, m_B  BR ,  m_T  TB 
\end{bmatrix}\tran\,.
\end{equation*}

$\mathbf{A_2}$ is upper diagonal, so its eigenvalues are $ \lambda_{21}\left(c_2\right), \; \lambda_{22}\,, \; \lambda_{23} \text{ and }  \lambda_{24}\left(c_2\right)$
with corresponding eigenvectors
  \begin{align}
     \mathbf{u}_{21} & = \begin{bmatrix}
     \frac{a\gamma\kappa c_2 m_B}{b \theta_R^2 \sigma c_1^2}, -\frac{a^2\gamma\kappa m_B}{\theta_R^2\sigma c_1^2},\frac{\kappa m_B}{\theta _R^2},1 
     \end{bmatrix}\tran\,, \label{eigenvecE2}\\
     \mathbf{u}_{22} & = \begin{bmatrix}
      \frac{a^2 c_2 \gamma \left(c_1-c_2\right)}{bc_1\left(a\left(c_1-c_2\right)-\theta_Rc_1\right)\left(\sigma c_1-a\theta_R\right)},\frac{a^2\gamma}{c_1\left(a\theta_R -\sigma c_1\right)},1,0
     \end{bmatrix}\tran\,, \nonumber\\
     \mathbf{u}_{23} & = \begin{bmatrix}
    -\frac{a c_2\left(c_1-c_2\right)}{b\left(a^2\left(c_1-c_2\right)-\sigma c_1^2\right)},1,0,0
     \end{bmatrix}\tran \quad \text{ and } 
     \mathbf{u}_{24}  = \begin{bmatrix}
     1,0,0,0
     \end{bmatrix}\tran\,.\nonumber
 \end{align}

 We transform system \eqref{trSys2} to its normal form with the help of the transformation
 \begin{equation} \label{transNormal2}
     \mathbf{X} = \mathbf{U_2} \mathbf{Y}\,,
 \end{equation}
 where we let 
  \begin{equation*}
    \mathbf{U_2} =  \begin{bmatrix}
      \mathbf{u}_{21}, \mathbf{u}_{22}, \mathbf{u}_{23}, \mathbf{u}_{24}
     \end{bmatrix}\,.
 \end{equation*}
By replacing \eqref{transNormal2} to system \eqref{trSys2}, we get
 \begin{equation}
      \dv{\mathbf{Y}}{t}   = \mathbf{\Lambda_2 Y} + \mathbf{F_2}\,,           \label{NormalForm2}
 \end{equation}
 where 
 \begin{equation*}
    \mathbf{\Lambda_2} = \text{diag}(\lambda_{21}\left(c_2\right),\lambda_{22},\lambda_{23},\lambda_{24}\left(c_2\right)) \quad \text{and} \quad
    \mathbf{F_2} = \begin{bmatrix}
      f_{21}, f_{22}, f_{23}, f_{24}
     \end{bmatrix}\tran = \mathbf{U_2^{-1} P_2}\,.
 \end{equation*}

By using equation \eqref{transNormal2} and substituting $T, N, R$ and  $B$ to $\mathbf{F_2}$,  we find that
\begin{align*}
     f_{21} &= w_{21} y_1^2 + w_{22} y_1 y_2 + w_{23} y_1 y_3 - m_T y_1 y_4 \,,\\
 f_{22} &= w_{24} y_1^2 + w_{25} y_1 y_2 + w_{26} y_1 y_3 - w_{27} y_1 y_4 \,,\\
 f_{23} &= w_{28} y_1^2 + w_{29} y_2^2 + \gamma y_2 y_3 + w_{210} y_1 y_2 + w_{211} y_1 y_3 - w_{212} y_1 y_4 \,,\\
  f_{24} &=   \begin{multlined}[t]
  w_{213} y_1^2 + w_{214} y_2^2 + w_{215} y_3^2 + w_{216} y_3 y_4 + a b y_4^2 + w_{217} y_2 y_3  \\ +w_{216} y_2 y_4 + w_{217} y_1 y_2 + w_{218} y_1 y_3 - w_{219} y_1 y_4\,,
 \end{multlined}
\end{align*}
 where $w_{2i}, \; i=1,2,3,...,19$ are known constants.

 Hence, system \eqref{NormalForm2} can be written in the form
 \begin{equation*}
     \dv{\mathbf{Y}}{t} = \begin{bmatrix}
    \mathbf{B_2} & \mathbf{0} \\
   \mathbf{0} & \mathbf{C_2}   \\
\end{bmatrix} \mathbf{Y} + \mathbf{F_2}\,,
 \end{equation*}
 or equivalently 
 \begin{subequations}
  \begin{align}
      \dv{y_1}{t} & = \mathbf{B_2} y_1 + f_{21}   \,,    \label{y1'2} \\   
    \begin{bmatrix}
        \dv{y_2}{t} \\ \dv{y_3}{t} \\ \dv{y_4}{t}
    \end{bmatrix}  
    & =
   \mathbf{C_2} \begin{bmatrix}
        y_2 \\ y_3 \\ y_4
    \end{bmatrix} 
     + \begin{bmatrix}
        f_{22} \\ f_{23} \\ f_{24}
    \end{bmatrix}\,,
  \end{align}
 \end{subequations}
 where
 \begin{equation*}
     \mathbf{B_2} = \lambda_{21}\left(c_2\right) = 0 \quad \text{ and } \quad \mathbf{C_2} = \text{diag}(\lambda_{22},\lambda_{23},\lambda_{24}\left(c_2\right))\,.
 \end{equation*}

We have that $\mathbf{B_2} $ and $\mathbf{C_2}$ are constant matrices, with the eigenvalues of $\mathbf{B_2} $ having 0 real part and the eigenvalues of $\mathbf{C_2}$ having negative real part, whereas $f_{2i}$ is smooth with $f_{2i}(0,0,0,0) = 0$ and $Df_{2i} (0,0,0,0) = 0$ for $i=1,2,3,4.$
By the center manifold theorem, there exists a center manifold which is parametrised by 
 \begin{equation*}
     \begin{bmatrix}
        y_2 \\ y_3 \\ y_4
    \end{bmatrix}  = \mathbf{h_2}(y_1) = \begin{bmatrix}
        h_{22}(y_1) \\ h_{23}(y_1) \\ h_{24}(y_1)
    \end{bmatrix}\,,
 \end{equation*}
  with $\mathbf{h_2}(0)=0$ and $D\mathbf{h_2}(0) = 0$, and satisfying 
 \begin{equation*}
     \mathbf{C_2} \cdot \mathbf{h_2}(y_1) +  \begin{bmatrix}
        f_{22}(y_1, \mathbf{h_2}(y_1) ) \\ f_{23}(y_1, \mathbf{h_2}(y_1) ) \\ f_{24}(y_1, \mathbf{h_2}(y_1) ) 
    \end{bmatrix} = D\mathbf{h_2}(y_1) \cdot \left[\mathbf{B_2}y_1 + f_{21}(y_1, \mathbf{h_2}(y_1)) \right]\,,
 \end{equation*}
while the flow on the center manifold is defined by the differential equation \eqref{y1'2}.

 Since $\mathbf{h_2}(0)=0$ and $D\mathbf{h_2}(0) = 0$, then by approximating the center manifold with a Taylor series around 0, we get
\begin{equation} \label{taylorApprox2}
    \mathbf{h_2}(y_1) = \begin{bmatrix}
        b_{22} y_1^2 + b_{23} y_1^3 + b_{24} y_1^4 + O(y_1^5) \\ c_{22} y_1^2 + c_{23} y_1^3 + c_{24} y_1^4 + O(y_1^5) \\ d_{22} y_1^2 + d_{23} y_1^3 + d_{24} y_1^4 + O(y_1^5)
    \end{bmatrix}\,.
\end{equation}
By replacing \eqref{taylorApprox2} to \eqref{y1'2} in order to express the terms of $y_2, y_3, y_4$ in $f_{21}$, we get

\begin{equation} \label{flowT2}
    \dv{y_1}{t} = w_{21}y_1^2 + O(y_1^3) \,,
\end{equation}
 with
 \begin{equation*}
     w_{21} = \frac{c_2 \gamma  \kappa \sigma  m_B m_T}{a b \left(\gamma  k+\theta _N \theta _R\right){}^2} > 0 \,.
 \end{equation*}
 
 From  $\mathbf{Y} = \mathbf{U_2^{-1} X}$, we see that $y_1 = B$, so $y_1$ is non-negative. That means that 0 is an unstable equilibrium point of equation \eqref{flowT2}, when approached from $y_1 \ge 0$ since $w_{21} > 0$. Thus, from the center manifold theorem, $\mathbf{0}$ is an unstable equilibrium point of system \eqref{trSys2}, which proves the proposition.
 
\end{proof}

Having studied every possible case in which the equilibrium points are biologically realistic, we conclude with the main result of the present section as follows. 
\begin{secproposition}[Local stability analysis]
\label{summain}
 For the Admissible equilibrium points of system \eqref{reducedModel}, the following facts hold true. 
  \begin{itemize}
    \item $E_1$ always exists, and is locally stable when $c \ge c_1$, and unstable when $c < c_1$.
    \item $E_2$ exists when $c < c_1$, and is locally stable when $c \ge c_2$, and unstable when $c < c_2$.
    \item $E_3$ exists when $c > c_2$ and $\frac{\theta_B}{m_T} < \frac{1}{b}$, and is always unstable.
    \item $E_1 \equiv E_2$ when $c = c_1$, and $E_2 \equiv E_3$ when $c = c_2$.
  \end{itemize}
\end{secproposition}

\section{Bifurcation Analysis} \label{BifAn}
Starting off, we prove that, under certain sufficient conditions,  IVP has no closed orbits and therefore a Hopf bifurcation can not occur.

\begin{secproposition}[Lack of Hopf bifurcations]
If $$0\leq T_0\leq \frac{1}{b}\text{ and }\frac{\theta_B}{m_T} > \frac{1}{b},$$ then IVP  has no closed orbits.
\end{secproposition}

\begin{proof}

From \eqref{dB} along with the non-negativity of the solution (see \hyperref[unnneg]{Proposition \ref*{unnneg}}) we get that $$\mathrm{sgn}{\frac{\de B}{\de t}}=\mathrm{sgn}{\left(T-\frac{\theta_B}{m_T}\right)}.$$ Since $T_0\leq 1/b$, we have from \hyperref[BofT]{Proposition \ref*{BofT}} that $T\leq 1/b$, as well. Hence, by $$\frac{\theta_B}{m_T} > \frac{1}{b}$$ we get that $B$ is strictly decreasing, therefore there are no periodic solutions of non-zero period, i.e no closed orbits.
\end{proof}

\begin{secremark}
A different way to prove the lack of Hopf bifurcation is by  supposing that there is a periodic solution $\left(T,N,R,B\right)$ of period $P\neq 0$, hence every component is also a $P$-periodic function. Then, from \eqref{Bfor} we get that $e^{\int_0^t  \left( m_TT(s) - \theta_B \right) \, \mathrm{d}s}$ is also $P$-periodic, and so does $$\int_0^t  \left( T(s) - \frac{\theta_B}{m_T} \right) \, \mathrm{d}s.$$ From this, along with the fact that $m_TT-\theta_B$ is $P$-periodic, we get that $$\int_0^P  \left( T(s) - \frac{\theta_B}{m_T} \right) \, \mathrm{d}s=0.$$ Dealing as before, we have a contradiction, since $$T-\frac{\theta_B}{m_T}\leq \frac{1}{b}-\frac{\theta_B}{m_T}<0.$$
\end{secremark}

The result of \hyperref[summain]{Proposition \ref*{summain}} makes us suspect that two transcritical bifurcations are happening; one between $E_1$ and $E_2$ when $c=c_1$, and one between $E_2$ and $E_3$ when $c=c_2$. Indeed, we set the right-hand side of system \eqref{reducedModel} equal to $\mathbf{G}$, i.e.
  \begin{equation*}
  \mathbf{G}(T,N,R,B) \coloneqq
  \begin{bmatrix}
      aT(1-bT) - cNT \\
      \sigma - \theta_NN - \gamma RN  \\
      \kappa -\theta_RR + m_BBR   \\
       -\theta_BB + m_TTB
  \end{bmatrix}\,,
  \end{equation*}
and we use Sotomayor's theorem \cite{perko2006differential} to prove our observations. 

\begin{secproposition}[Transcritical bifurcation 1]
  System \eqref{reducedModel} experiences a transcritical bifurcation at the equilibrium point $E_1 \equiv E_2$ as the parameter $c$ varies through $c_1$.
\end{secproposition}
\begin{proof}
 The Jacobian matrix of system \eqref{reducedModel} at $E_1 \equiv E_2$, i.e. when $c = c_1$, is given by $\mathbf{A_1}$. From \eqref{eigenvalE1} and \eqref{eigenvecE1.a}, we see that the eigenvector corresponding to the zero eigenvalue is $\mathbf{u}_{11}$, whereas simple computations show us that the eigenvector corresponding to the zero eigenvalue of matrix $\mathbf{A_1}\tran$ is $\mathbf{w_1} = \begin{bmatrix}
      1,0,0,0 
     \end{bmatrix}\tran$.
  
  Furthermore, we have that
  \begin{equation*}
   \pdv{\mathbf{G}(T,N,R,B;c)}{c} = 
   \begin{bmatrix}
      -NT,  0,  0  ,  0
  \end{bmatrix}\tran\,,
  \end{equation*}
  
    \begin{equation*}
   D\pdv{\mathbf{G}(T,N,R,B;c)}{c}\mathbf{u}_{11} = 
   \begin{bmatrix}
      -N,  0,  0  ,  0
  \end{bmatrix}\tran\,,
  \end{equation*}
  and
   \begin{equation*}
   D^2 \mathbf{G}(T,N,R,B) (\mathbf{u}_{11},\mathbf{u}_{11}) = 
   \begin{bmatrix}
      -2ab,  0,  0  ,  0
  \end{bmatrix}\tran\,.
  \end{equation*}
  
  Consequently, when $c=c_1$ we have that
  \begin{align*}
      & \mathbf{w_1}\tran \pdv{\mathbf{G}(E_1;c_1)}{c} =0\,, \\
      & \mathbf{w_1}\tran \left[ D\pdv{\mathbf{G}(E_1;c_1)}{c}\mathbf{u}_{11} \right] = -\frac{a}{c_1} \neq 0\,, \\
      & \mathbf{w_1}\tran \left[ D^2 \mathbf{G}(E_1;c_1) (\mathbf{u}_{11},\mathbf{u}_{11})\right] = -2 a b \neq 0\,.
  \end{align*}
  
  Thus, the conditions of Sotomayor's theorem are satisfied and the proposition is proved.
\end{proof}

\begin{secproposition}[Transcritical bifurcation 2]
  System \eqref{reducedModel} experiences a transcritical bifurcation at the equilibrium point $E_2 \equiv E_3$ as the parameter $c$ varies through $c_2$.
\end{secproposition}

\begin{proof}
  Since 
  \begin{equation*}
      c = c_2 \Leftrightarrow \theta_B = \frac{a \theta _N m_T \theta _R+a \gamma  \kappa  m_T-c \sigma  m_T \theta _R}{a b \gamma  \kappa +a b \theta _N \theta _R} =: \theta_B^*\,,
  \end{equation*}
  we prove the equivalent proposition of system \eqref{reducedModel} experiencing a transcritical bifurcation at the equilibrium point $E_2 \equiv E_3$ as the parameter $\theta_B$ varies through $\theta_B^*$, as the conditions of Sotomayor's theorem cannot be satisfied in the case of parameter $c$ varying through $c_2$.
  
  We have that the Jacobian of system \eqref{reducedModel} at $E_2 \equiv E_3$, i.e. when $\theta_B = \theta_B^*$, is equal to $   \mathbf{A_2}=\mathbf{J}(E_2;c=c_2)=\mathbf{J}(E_2;\theta_B=\theta_B^*)$.
 
From \eqref{eigenvalE2} and \eqref{eigenvecE2}, we see that the eigenvector corresponding to the zero eigenvalue is $\mathbf{u}_{21}$, whereas simple computations show us that the eigenvector corresponding to the zero eigenvalue of matrix $\mathbf{A_2}\tran$ is $\mathbf{w_2} = \begin{bmatrix}
       0, 0, 0, 1 
     \end{bmatrix}\tran$.

Furthermore, we have that
  \begin{equation*}
   \pdv{\mathbf{G}(T,N,R,B;\theta_B)}{\theta_B} = 
   \begin{bmatrix}
      0,0,0,-B
  \end{bmatrix}\tran,
  \end{equation*}
    \begin{equation*}
   D\pdv{\mathbf{G}(T,N,R,B;\theta_B)}{\theta_B}\mathbf{u}_{21} = 
   \begin{bmatrix}
      0, 0, 0, -1
  \end{bmatrix}\tran\,,
  \end{equation*}
  and
   \begin{equation*}
   D^2 \mathbf{G}(T,N,R,B;\theta_B) (\mathbf{u}_{21},\mathbf{u}_{21}) = 
   \begin{bmatrix}
     0,\frac{2 \gamma ^2 \kappa ^2 \sigma  m_B^2}{\theta _R^2 \left(\gamma  \kappa +\theta _N \theta _R\right){}^2},\frac{2 \kappa  m_B^2}{\theta _R^2},\frac{2 c \gamma  \kappa  \sigma  m_B m_T}{a b \left(\gamma  \kappa +\theta _N \theta _R\right){}^2}
  \end{bmatrix}\tran\,.
  \end{equation*}

Consequently, when $\theta_B = \theta_B^*$ we have that
  \begin{align*}
      & \mathbf{w_2}\tran \pdv{\mathbf{G}(E_2;\theta_B^*)}{\theta_B} =0\,, \\
      & \mathbf{w_2}\tran \left[ D\pdv{\mathbf{G}(E_2;\theta_B^*)}{\theta_B} \right] = -1 \neq 0\,, \\
      & \mathbf{w_2}\tran \left[ D^2 \mathbf{G}(E_2;\theta_B^*) (\mathbf{u}_{21},\mathbf{u}_{21})\right] = \frac{2 c \gamma  \kappa  \sigma  m_B m_T}{a b \left(\gamma  \kappa +\theta _N \theta _R\right){}^2} \neq 0\,.
  \end{align*}
  
  Thus, the conditions of Sotomayor's theorem are satisfied and the proposition is proved.

\end{proof}

The value and stability of the first coordinate - that is the breast cancer cell population - at equilibrium, depending on parameter $c$ is given in \hyperref[Fig3]{Figure \ref*{Fig3}}.
%%%%%%%%%%%%%%%%%%%%%%%
\begin{figure}[!htbp]
\begin{subfigure}{.5\textwidth}
  \centering
\includegraphics[width=1\linewidth]{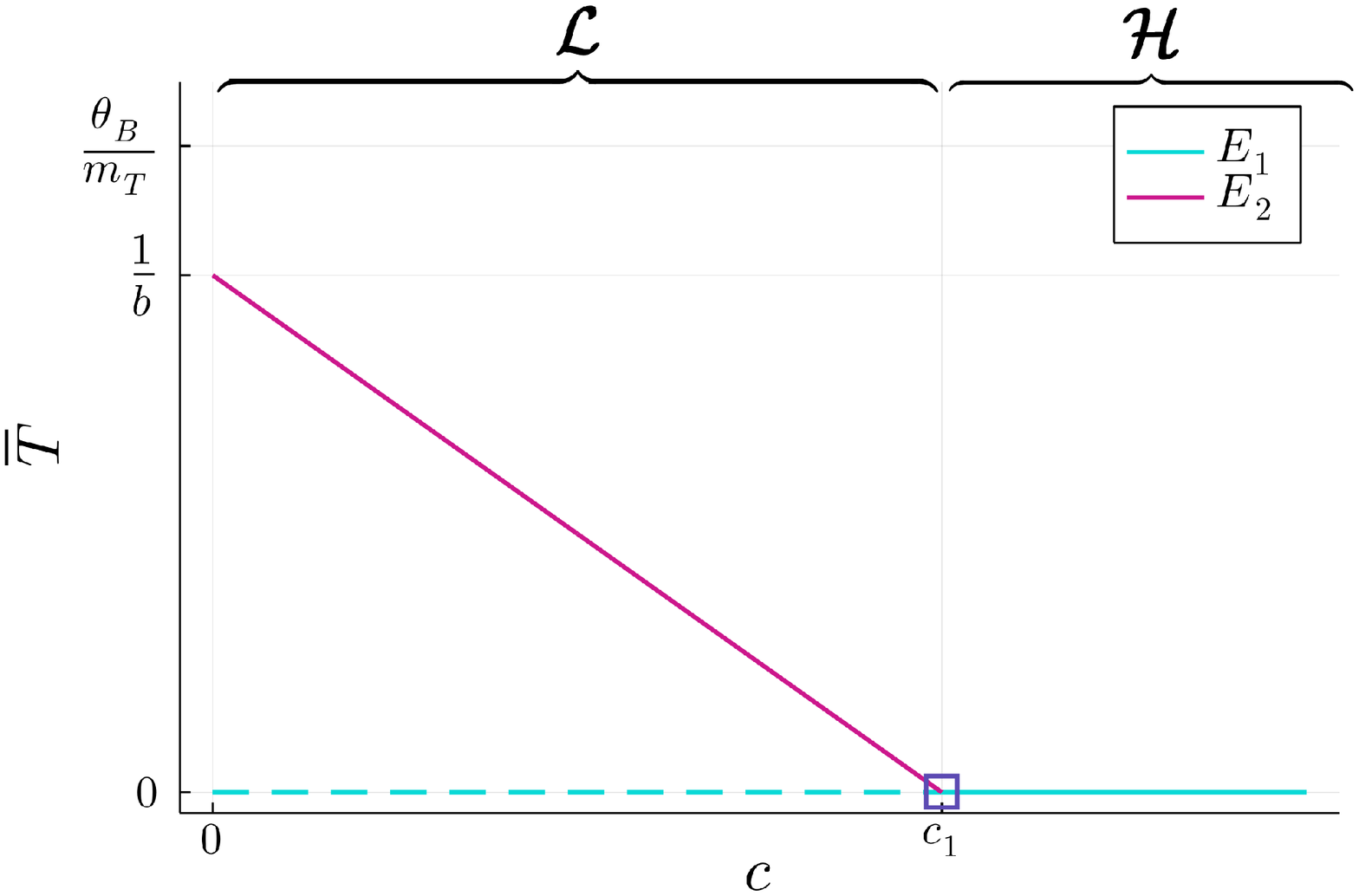}
 \caption{The tumour carrying capacity is bounded from above, i.e. $\frac{\theta_B}{m_T}\geq \frac{1}{b}$.}
  \label{Fig3.a}
\end{subfigure}
\begin{subfigure}{.5\textwidth}
  \centering
 \includegraphics[width=1\linewidth]{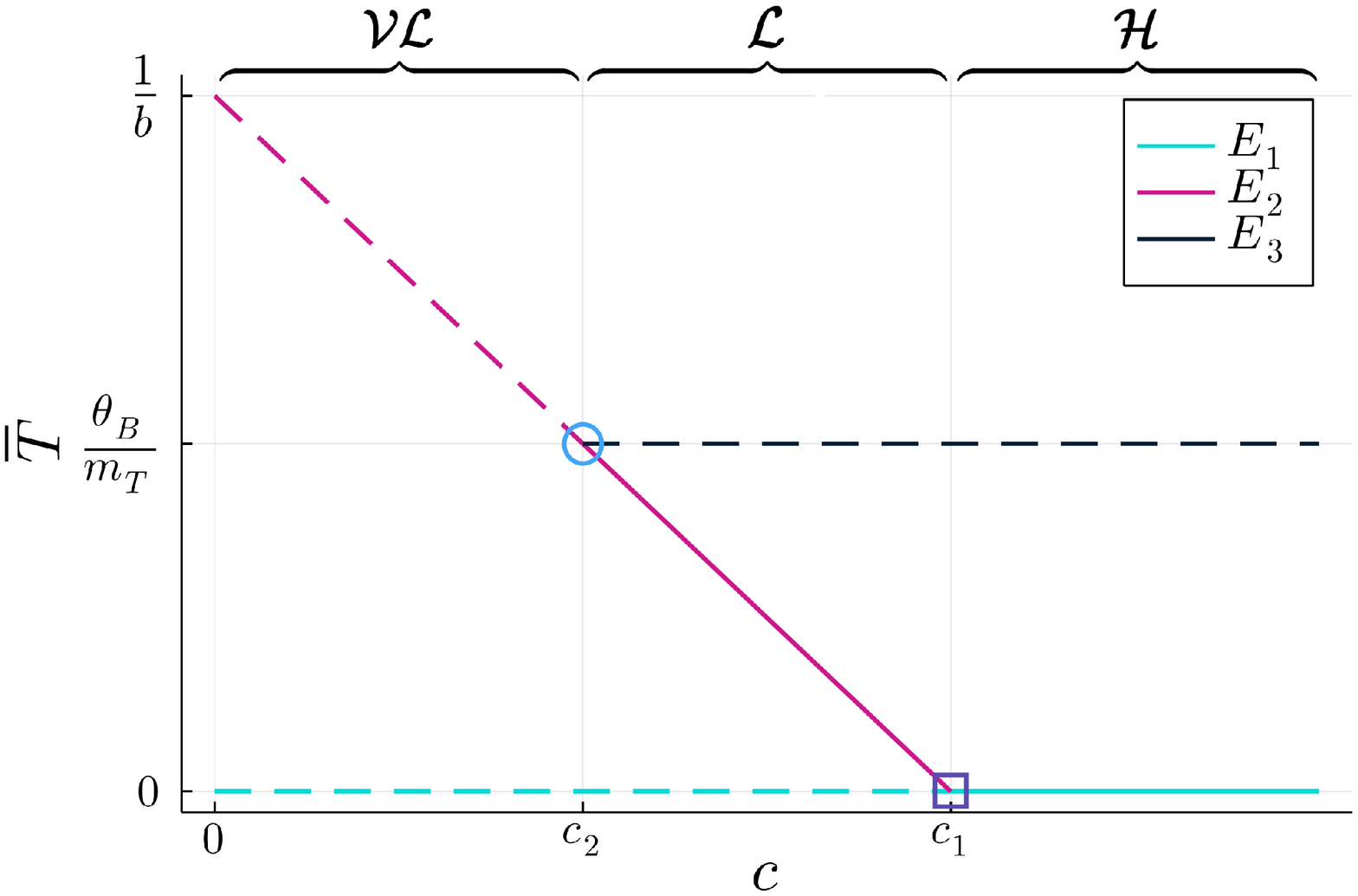}
  \caption{The tumour carrying capacity is bounded from below, i.e. $\frac{\theta_B}{m_T}<\frac{1}{b}$.}
  \label{Fig3.b}
\end{subfigure}
\caption{Projection of the bifurcation diagram of system \eqref{reducedModel} onto the $\overbar{T}$--$c$~plane. Solid line (---): Stable equilibrium. Dashed line (- -): Unstable equilibrium. $\square$: Transcritical bifurcation 1. $\bigcirc$: Transcritical bifurcation 2. $\mathcal{VL}$: Region in which $c$ is \textquote{very low}, i.e. $c<c_2$ when $c_2>0$. $\mathcal{L}$: Region in which $c$ is \textquote{low}, i.e. $c<c_1$ when $c_2<0$ and $c_2<c<c_1$ when $c_2>0$. $\mathcal{H}$: Region in which $c$ is \textquote{high}, i.e. $c>c_1$.}
   \label{Fig3}
\end{figure}

%%%%%%%%%%%%%%%%%%%%%%%

\section{Numerical Simulations} \label{reducedModelNumericalSimulationsSection}

In this section, we numerically solve our model, using Julia and the suite DifferentialEquations.jl \cite{rackauckas2017differentialequations}, under two antipodal theoretical scenarios concerning the absence and the presence of tBregs. Contrary to our previous local stability analysis, numerical simulations help us develop an intuition about the global behaviour of our model. Additionally, we perform numerical parameter sensitivity analysis, in order to further explore the dependence of the solution of our model to its parameters, under different initial conditions of tBregs.

The parameter values used in our simulations are listed in \hyperref[reducedModelParameterTable]{Table \ref*{reducedModelParameterTable}}, unless otherwise stated, and their derivation is explained in \hyperref[secB1]{Appendix}. In particular, $m_T$ either becomes
\begin{equation*}
    m_T = m_T^A \coloneqq  5 \cdot 10^{-15}\; \text{cell}^{-1}  \, \cdot \, \text{day}^{-1} \;,
\end{equation*} or
    
\begin{equation*}
    m_T = m_T^B \coloneqq  5 \cdot 10^{-10}\; \text{cell}^{-1}  \, \cdot \, \text{day}^{-1} \; .
\end{equation*}
Thus, choosing the value $m_T^A$ yields $\frac{\theta_B}{m_T} > \frac{1}{b}$, whereas choosing the value $m_T^B$ yields $\frac{\theta_B}{m_T} < \frac{1}{b}$, allowing us to explore the different dynamics of IVP.

Considering the values of two bifurcation points of system \eqref{reducedModel} with respect to parameter $c$, we have that $c_1 = 2.88 \cdot 10^{-10}\; \text{cell}^{-1}  \, \cdot \, \text{day}^{-1}$, since $c_1$ is independent from the value of $m_T$. When $m_T = m_T^A$, then $c_2<0$, as expected from our bifurcation analysis, whereas when $m_T = m_T^B$, then $c_2 = 5.76 \cdot 10^{-11}\; \text{cell}^{-1}  \, \cdot \, \text{day}^{-1}$.

Finally, we choose $N_0 = 5 \cdot 10^8$ cells and $R_0 = 2 \cdot 10^8$ cells.

{\footnotesize
\begin{table}[h]
\caption{Units and values of the parameters of system \eqref{reducedModel}. For estimations arising from literature see \hyperref[secB1]{Appendix \ref*{secB1}}.}
\label{reducedModelParameterTable}
\noindent\makebox[\textwidth]{
\begin{tabular}{p{.5cm}p{2.5cm}p{2cm}p{5cm}}
\toprule
\textbf{Prm.} &    \textbf{Unit} & \textbf{Value}  & \textbf{Source}    \\ 
\midrule
$a$    & day$^{-1}$      & 0.15    & \cite{bitsouni2021mathematical}\\
$b$  & cell$^{-1}$  & $1 \cdot 10^{-9}$ & \cite{bitsouni2021mathematical}  \\
$c$ & cell$^{-1} \cdot$ day$^{-1}$ & varied                 & \begin{center} --- \end{center} \\ 
\midrule
$\sigma$    & cell $\cdot$ day$^{-1}$           & $5 \cdot 10^7$    & Estimated from \cite{zhang}    \\
$\theta_N$        & day$^{-1}$                        & 0.07              & Estimated from \cite{zhang}     \\
$\gamma$      & cell$^{-1} \cdot$ day$^{-1}$ & $1 \cdot 10^{-10}$     & Estimated from equilibrium point\\ 
\midrule
$\kappa$        & cell $\cdot$ day$^{-1}$           & $1 \cdot 10^7$    & Estimated from equilibrium point \\
$\theta_R$             & day$^{-1}$                        & 0.03851           & Estimated from \cite{mabarrack2008}  \\
$m_B$        & cell$^{-1} \cdot$ day$^{-1}$ & $3 \cdot 10^{-8}$      & No data found  \\
\midrule
$\theta_B$      & day$^{-1}$                        & 0.4               & No data found\\
$m_T$   & cell$^{-1} \cdot$ day$^{-1}$ & varied                 & \begin{center} --- \end{center} \\
\bottomrule
\end{tabular}}
\end{table}
}

\subsection{The Scenario of Absent tBregs ($B_0 = 0$)}
As can be seen from \eqref{Bfor}, when $B_0 = 0$, we have that $B(t) = 0, \; t \in \mathbb{R}_{\ge 0}$. Therefore, system \eqref{reducedModel} becomes a 3D system with only breast cancer cells, NK cells and Tregs existing in the body.

Plotting the phase portrait of the resulting 3D system, for different values of $c$, we get \hyperref[Fig4]{Figure \ref*{Fig4}}. We notice that all trajectories move towards an equilibrium point for all $c$. This equilibrium point moves along the line $$\{(T,N,R) \in \mathbb{R}_{\ge 0 }^3: N=\frac{\sigma  \theta _R}{\gamma  \kappa +\theta _N \theta _R}, R= \frac{\kappa }{\theta _R}\},$$ with $T$ getting bigger, as $c$ gets smaller, and vice versa. When $c>c_1$, T reaches zero. We observe the same type of transcritical bifurcation that happens between $E_1$ and $E_2$ when $c=c_1$ for system \eqref{reducedModel}, also happens for the 3D system. Even with no rigorous result at hand regarding the global stability of IVP with tBregs no longer in the picture, it is apparent that for the 3D system the equilibrium point $E_1$ is globally stable when $c>c_1$ and $E_2$ is globally stable when $c<c_1$. The crux of the scenario in question is:

\noindent\fbox{\parbox{\textwidth}{\begin{center}\textbf{Conclusion $1_a$:\\} In the absence of tBregs, the breast tumour will reach its carrying capacity due to NK cell insufficiency, i.e. $T \nearrow \frac{1}{b}$ when $c \searrow 0$.  \end{center}}} \label{ConclusionA}
%%%%%%%
\begin{figure}[h!]
 \makebox[\textwidth][c]{\includegraphics[width=1\textwidth]{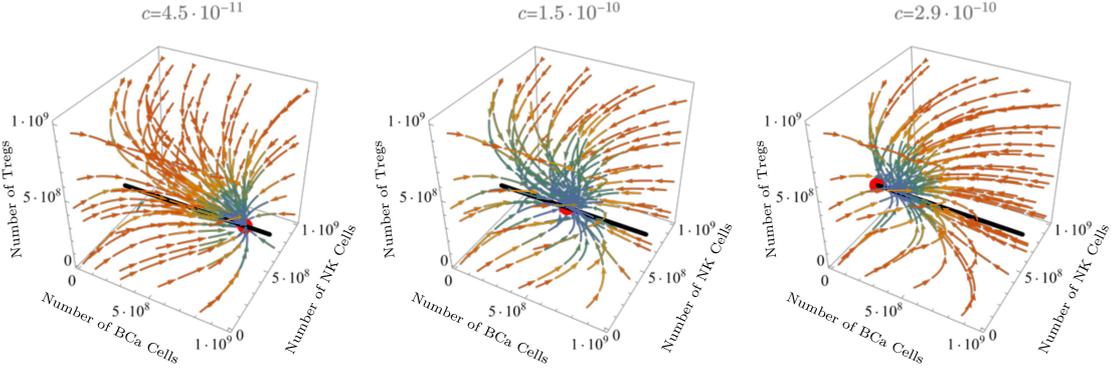}}%
  \caption{Phase portrait of system \eqref{reducedModel} with $B_0 = 0$, for different values of $c$. The equilibrium point $(\overbar{T},\overbar{N},\overbar{R})$ pictured by the red sphere moves along the black line, which is given by $\{(T,N,R) \in \mathbb{R}_{\ge 0 }^3: N=\frac{\sigma  \theta _R}{\gamma  \kappa +\theta _N \theta _R}, R= \frac{\kappa }{\theta _R}\}$. The value of $\overbar{T}$ tends to 0, as $c$ increases, whereas tends to $\frac{1}{b}$, as $c$ tends to 0.}
  \label{Fig4}
\end{figure}
%%%%%%%%
\subsection{The Scenario of Present tBregs ($B_0 \neq 0$)} \label{sec:tBregsExist}
In this section, we study the scenario in which the initial condition of tBregs is not 0. We divide our analysis into three sections: \hyperref[sec:GE1b]{Section \ref*{sec:GE1b}} in which $m_T = m_T^A$ and therefore $\frac{\theta_B}{m_T} \ge \frac{1}{b}$, and \hyperref[sec:L1b]{Section \ref*{sec:L1b}} in which $m_T = m_T^B$ and therefore $\frac{\theta_B}{m_T} < \frac{1}{b}$. In both first sections, we individually study each of the involved regions, $\mathcal{VL}$, $\mathcal{L}$ and $\mathcal{H}$, admitted by our bifurcation analysis and depicted in \hyperref[Fig3]{Figure \ref*{Fig3}}. In \hyperref[sec:importc]{Section \ref*{sec:importc}} we investigate the influence of  the rate of NK-induced breast cancer cell death, $c$, to our model.

\subsubsection{Bounded-From-Above Tumour Carrying Capacity ($\frac{\theta_B}{m_T} \ge \frac{1}{b}$)} \label{sec:GE1b}
Here, we show that the dynamics of the solution of IVP are eventually the same with the ones of the previous scenario. However, there are interesting initial differences between the two scenarios. 

When $m_T = m_T^A$, we have that  $\frac{\theta_B}{m_T} \ge \frac{1}{b}$, which means that B is strictly decreasing. Additionally, only equilibria $E_1$ and $E_2$ exist. \hyperref[Fig5]{Figure \ref*{Fig5}} and \hyperref[Fig6]{Figure \ref*{Fig6}} depict both $T$ and $B$, for $c$ in region $\mathcal{L}$ and $\mathcal{H}$, respectively. As can be seen in the aforementioned figures, $T$ forms a single peak, which increases as $B_0$ increases. Such an increase of $T$ continues until $T$ reaches its carrying capacity, $\frac{1}{b}$, and after that it gradually flattens out. The greater $B_0$ is, the longer $T$ stays close to $\frac{1}{b}$, before eventually tending to $\overbar{T}$ of $E_2$. Consequently, as $B_0$ increases, the metastatic potential becomes higher. This comes in agreement with the results of \cite{olkhanud2011tumor}, where tBregs were associated with breast cancer metastasising to the lungs. 
%%%%%%
\begin{figure}[h!]
 \makebox[\textwidth][c]{\includegraphics[width=1\textwidth]{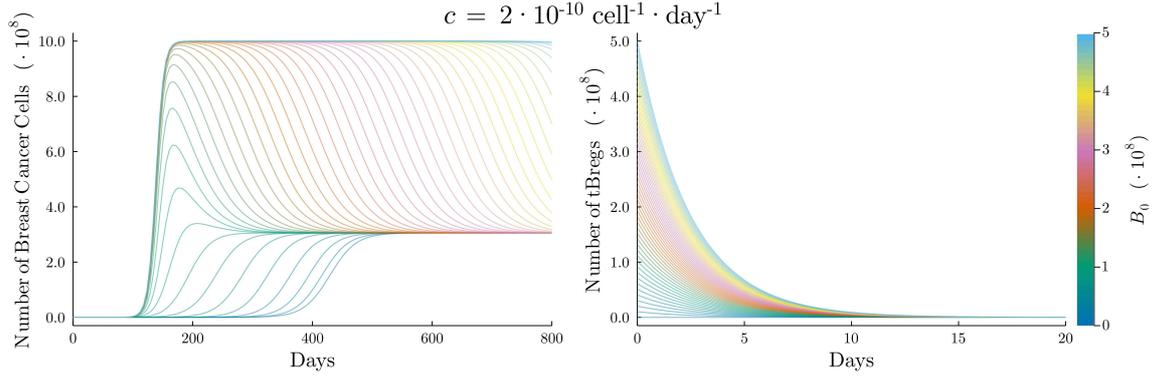}}%
    \caption{Ensemble simulations of IVP, with $B_0$ taking values in the 30-point discretisation of the interval $\left[ 0, 5\cdot 10^{8} \right] \, \cdot \, \text{cells}$,  $T_0 = 5$ cells, and $c=2 \cdot 10^{-10}\; \text{cell}^{-1}  \, \cdot \, \text{day}^{-1}$ (region $\mathcal{L}$). When $B\to 0$, $T$ eventually tends to $\overbar{T}>0$ of $E_2$.}
  \label{Fig5}
\end{figure}
%%%%%%%%
In \hyperref[Fig6]{Figure \ref*{Fig6}}, $T$ eventually vanishes, instead of tending to a positive value, as in \hyperref[Fig5]{Figure \ref*{Fig5}}. Thus, in the light of \hyperref[Fig4]{Figure \ref*{Fig4}}, when $B\to0$, the dynamics of IVP are similar to the corresponding ones of the problem with the simplified 3D system for $B=0$.
%%%%%%%% 
\begin{figure}[h!]
 \makebox[\textwidth][c]{\includegraphics[width=1\textwidth]{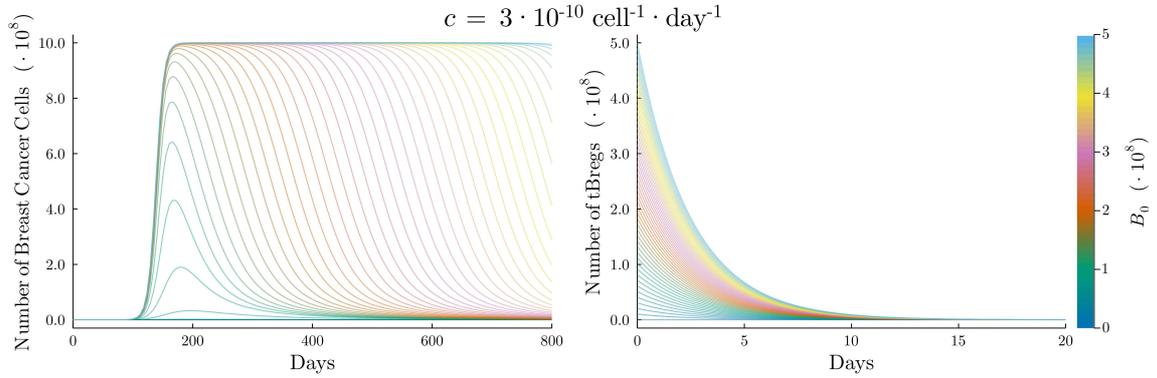}}%
      \caption{Ensemble simulations of IVP, with $B_0$ taking values in the 30-point discretisation of the interval $\left[ 0, 5\cdot 10^{8} \right] \, \cdot \, \text{cells}$ ,  $T_0 = 5$ cells, and $c=3 \cdot 10^{-10}\; \text{cell}^{-1}  \, \cdot \, \text{day}^{-1}$ (region $\mathcal{H}$). When $B\to 0$, $T$ eventually tends to $\overbar{T}=0$ of $E_1$.}
  \label{Fig6}
\end{figure}
%%%%%%%%%

All in all, we reach the following conclusion: 

\noindent\fbox{\parbox{\textwidth}{\begin{center}\textbf{Conclusion 2:\\} In the presence of tBregs, if tumour carrying capacity is bounded from above, then  the breast tumour will initially reach its carrying capacity due to the effect of tBregs, i.e. initially $T \nearrow \frac{1}{b}$ when $B_0 \nearrow $. \end{center}}} \label{conclusionB}

\subsubsection{Bounded-From-Bellow Tumour Carrying Capacity ($\frac{\theta_B}{m_T} < \frac{1}{b}$)} \label{sec:L1b}
 When $\frac{\theta_B}{m_T} < \frac{1}{b}$, $B$ can become an increasing function. Hence, even if tBregs start as just a few cells, they could cause breast cancer cells to reach their carrying capacity. 
 
\paragraph{Region $\mathcal{L}$} \label{sec:inBetween}
When $c=9 \cdot 10^{-11}\; \text{cell}^{-1}  \, \cdot \, \text{day}^{-1}$, we have that $c_2<c<c_1$. Depending on the initial condition of tBregs, the dynamics of IVP vary greatly. When $B_0 = 5 \cdot 10^5$ cells, we get \hyperref[Fig7]{Figure \ref*{Fig7}}. In \hyperref[Fig7]{Figure \ref*{Fig7}}, we see that tBregs, despite increasing up until around the 20th day, are unable to continue doing so and are eventually depleted. tBreg levels did not reach a sufficient population number that would allow breast cancer cells to grow to their carrying capacity. Instead, breast cancer cells stabilise to equilibrium point $E_2$.
%%%%%%
\begin{figure}[h!]
 \makebox[\textwidth][c]{\includegraphics[width=1\textwidth]{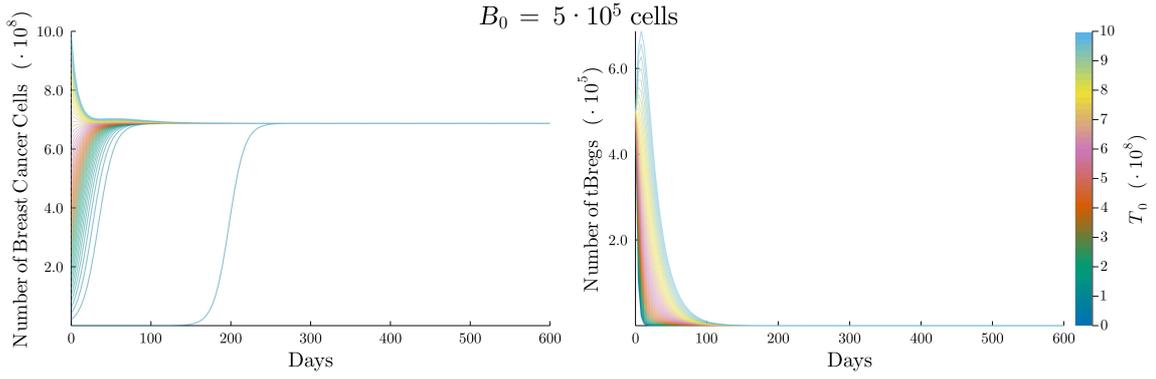}}%
      \caption{Ensemble simulations of IVP, with $T_0$ taking values in the 50-point discretisation of the interval $\left[ 0, 10^{9} \right] \, \cdot \, \text{cells}$, $c=9 \cdot 10^{-11}\; \text{cell}^{-1}  \, \cdot \, \text{day}^{-1}$ and  $B_0 = 5 \cdot 10^5$ cells. Left: Number of tBregs. Right: Number of breast cancer cells. We notice that when $B_0 = 5 \cdot 10^5$ cells, the trajectories of our model move towards $E_2$.}
  \label{Fig7}
\end{figure}
%%%%%
Increasing the initial condition of tBregs to $B_0 = 5 \cdot 10^6$ cells, we get \hyperref[Fig8]{Figure \ref*{Fig8}}. In \hyperref[Fig8]{Figure \ref*{Fig8}}, we see that when the initial breast cancer cell population is around $6 \cdot 10^8$ cells, tBregs do not get depleted. Additionally, in every simulation where tBregs do not deplete, they cause breast cancer cells to move away from equilibrium point $E_2$ and reach their carrying capacity.

\begin{figure}[h!]
 \makebox[\textwidth][c]{\includegraphics[width=1\textwidth]{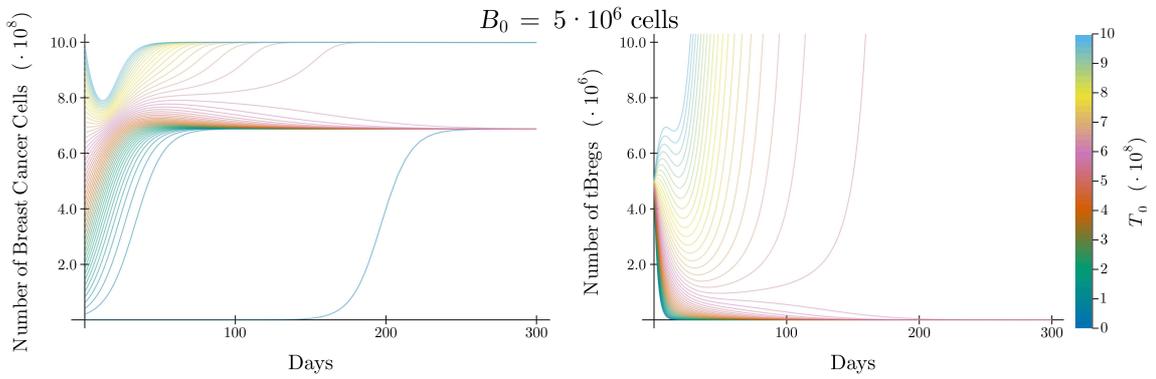}}%
     \caption{Ensemble simulations of IVP, with $T_0$ taking values in the 50-point discretisation of the interval $\left[ 0, 10^{9} \right] \, \cdot \, \text{cells}$, $c=9 \cdot 10^{-11}\; \text{cell}^{-1}  \, \cdot \, \text{day}^{-1}$ and  $B_0 = 5 \cdot 10^6$ cells. Left: Number of breast cancer cells. Right: Number of tBregs. We notice that, unlike the corresponding simulations of \hyperref[Fig7]{Figure \ref*{Fig7}}, when $B_0 = 5 \cdot 10^6$ cells, and if $T_0 \ge 6 \cdot 10^8$, then the breast tumour reaches its carrying capacity.}
  \label{Fig8}
\end{figure}

\paragraph{Region $\mathcal{VL}$} \label{sec:smaller}
When $c=2 \cdot 10^{-11}\; \text{cell}^{-1}  \, \cdot \, \text{day}^{-1}$, we have that $c<c_2$. Region $\mathcal{VL}$ is a particular interesting case, since our local stability analysis showed that both $E_1$ and $E_2$ are unstable, hence making the dynamics of IVP difficult to determine without numerical simulations. In  \hyperref[Fig9]{Figure \ref*{Fig9}}, we notice that despite the initial number of tBregs being as low as $B_0 = 50$ cells, tBregs manage to proliferate. Even though up until around the 200th day breast cancer cells seem to have stabilised at equilibrium point $E_2$, they proceed to start increasing up until they reach their carrying capacity. An interesting observation is that tBreg levels need to be around $10^8$ cells in order for breast cancer cells to reach their carrying capacity, much like the case studied for the scenario of $\frac{\theta_B}{m_T} \ge \frac{1}{b}$ in \hyperref[Fig6]{Figure \ref*{Fig6}}. 

\begin{figure}[h!]
 \makebox[\textwidth][c]{\includegraphics[width=1\textwidth]{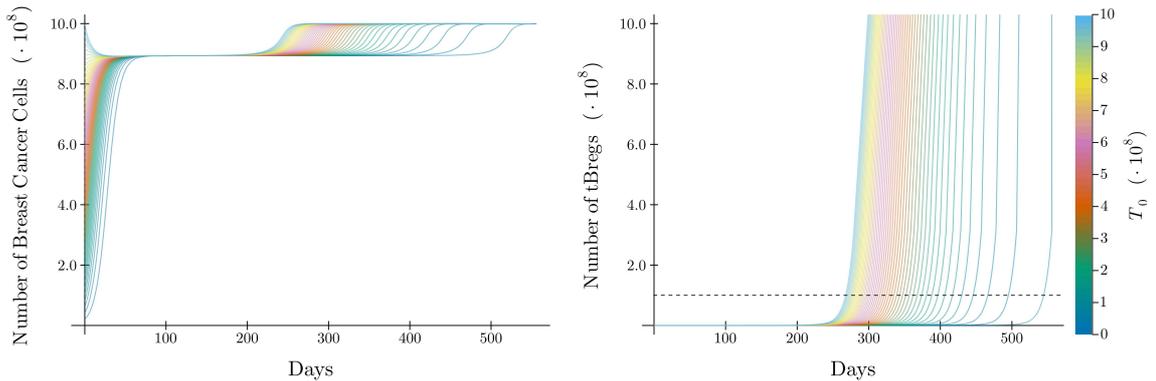}}%
     \caption{Ensemble simulations of IVP, with $T_0$ taking values in the 50-point discretisation of the interval $\left[ 0, 10^{9} \right] \, \cdot \, \text{cells}$, $c=2 \cdot 10^{-11}\; \text{cell}^{-1}  \, \cdot \, \text{day}^{-1}$ and  $B_0 = 50$ cells. Left: Number of breast cancer cells. Right: Number of tBregs. Even though breast cancer cells seem to have stabilised to their corresponding $E_2$ value for more than 100 days in each simulation, when tBerg levels are about $10^8$ cells (dashed line), breast cancer cells start increasing again and reach their carrying capacity.}
  \label{Fig9}
\end{figure}

Finally, we notice that when tBregs are scarce, even if the initial number of breast cancer cells is very close to the tumour carrying capacity, breast cancer cells are not able to maintain their high levels and decrease until they reach $E_2$. However, this changes when tBregs become around $10^8$ cells, causing breast cancer cells to increase to their carrying capacity, while also being able to maintain their high numbers.

\paragraph{Region $\mathcal{H}$} \label{sec:bigger}
When $c = 3 \cdot 10^{-10}\; \text{cell}^{-1}  \, \cdot \, \text{day}^{-1}$, we have that $c>c_1$. In this case, the ability of NK cells at lysing breast cancer cells is at its highest, so we expect some of our simulations to result in the elimination of the tumour. As anticipated, we see that in \hyperref[Fig10]{Figure \ref*{Fig10}}, when the initial condition of tBregs is $B_0 = 5 \cdot 10^{5}$, breast cancer cells, along with tBregs, are quickly depleted from the body for all simulations.

\begin{figure}[h!] 
 \makebox[\textwidth][c]{\includegraphics[width=1\textwidth]{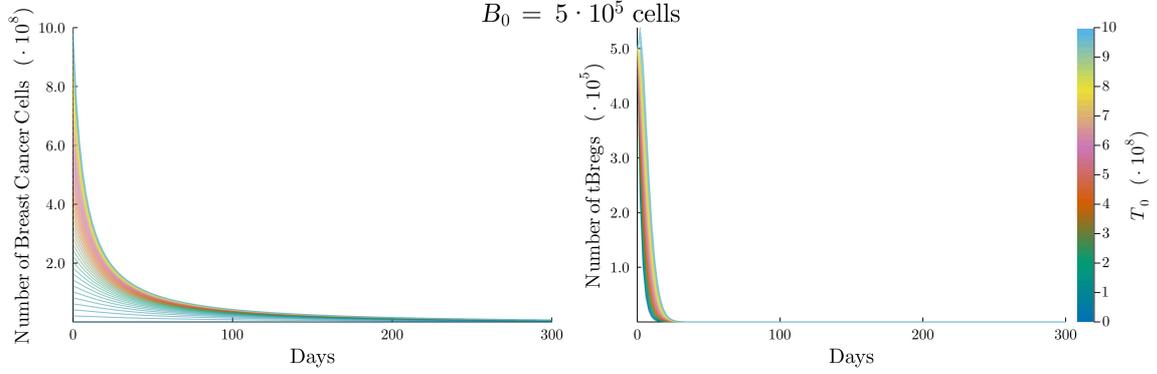}}%
     \caption{Ensemble simulations of IVP, with $T_0$ taking values in the 50-point discretisation of the interval $\left[ 0, 10^{9} \right] \, \cdot \, \text{cells}$, $c = 3 \cdot 10^{-10}\; \text{cell}^{-1}  \, \cdot \, \text{day}^{-1}$ and $B_0 = 5 \cdot 10^5$ cells. Left: Number of breast cancer cells. Right: Number of tBregs. We notice that breast cancer cells are quickly depleted from the body when $c = 3 \cdot 10^{-10}\; \text{cell}^{-1}  \, \cdot \, \text{day}^{-1} \; (c > c_1)$ and $B_0 = 5 \cdot 10^5$ cells.}
  \label{Fig10}
\end{figure}

However, the same does not hold if the initial number of tBergs is higher. When the initial condition of IVP becomes equal to equilibrium point $E_3$, which in this case is equal to $\left( 8 \cdot 10^8, 6 \cdot 10^7, 7.63 \cdot 10^9, 1.24 \cdot 10^6 \right) \; \cdot \; \text{cells}$, we get \hyperref[Fig11]{Figure \ref*{Fig11}}. In \hyperref[Fig11]{Figure \ref*{Fig11}}, we observe that when the initial condition of breast cancer cells is lower than the value of breast cancer cells at $E_3$, the tumour is eliminated, whereas if it is higher, the tumour reaches its carrying capacity. Naturally, if it is equal to the value of breast cancer cells at $E_3$, then the solution is constant. Even though the trajectories that move towards the tumour carrying capacity, do not actually move towards an equilibrium point, a sort of bistability phenomenon appears, as far as the population of breast cancer cells is concerned.

\begin{figure}[h!]
 \makebox[\textwidth][c]{\includegraphics[width=1\textwidth]{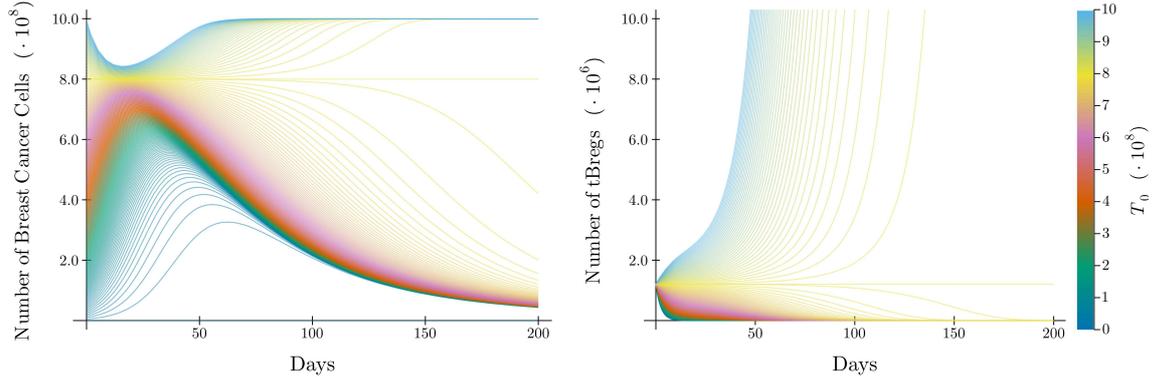}}%
     \caption{Ensemble simulations of IVP, with $T_0$ taking values in the 200-point discretisation of the interval $\left[ 0, 10^{9} \right] \, \cdot \, \text{cells}$, $c = 3 \cdot 10^{-10}\; \text{cell}^{-1}  \, \cdot \, \text{day}^{-1}$ and  $\left( T_0, N_0, R_0, B_0 \right) = E_3$ cells. Left: Number of breast cancer cells. Right: Number of tBregs.  We notice that if $T_0$ is higher than the corresponding value of $T$ in $E_3$ then $T \nearrow \frac{1}{b}$, whereas if $T_0$ is lower than the corresponding value of $T$ in $E_3$ then $T \searrow 0$. }
  \label{Fig11}
\end{figure}

Throughout \hyperref[sec:tBregsExist]{Section \ref*{sec:tBregsExist}}, we observed the importance of tBergs in the dynamics of IVP. Each time tBreg levels were sufficiently high, breast cancer cells reached their carrying capacity. In the presence of tBregs, breast tumour can reach its carrying capacity independently of the value of parameter $c$. Hence, we reach the following conclusion:
 
\noindent\fbox{\parbox{\textwidth}{\begin{center}\textbf{Conclusion 3:\\} In the presence of tBregs, if tumour carrying capacity is bounded from below, then  the breast tumour will eventually reach its carrying capacity due to the effect of an increasing population of tBregs, i.e. $T \nearrow \frac{1}{b}$ when $B \nearrow$.  \end{center}}} \label{conclusionC}

\subsubsection{The Importance of $c$}
\label{sec:importc}
Our bifurcation analysis, along with our numerical simulations, underlined the importance of the rate at which NK cells lyse breast cancer cells. A way to numerically visualise the bifurcation diagram of \hyperref[Fig3]{Figure \ref*{Fig3}} is with \hyperref[Fig12]{Figure \ref*{Fig12}}. In \hyperref[Fig12]{Figure \ref*{Fig12}}, we let $c$ take a range of values between  $4 \cdot 10 ^{-11}$ and $4 \cdot 10^{-10} \; \cdot \; \text{cell}^{-1}  \, \cdot \, \text{day}^{-1}$, while keeping the initial condition constant. As expected from our stability and numerical analysis, when $c<c_1$, breast cancer cells increase until they reach their carrying capacity. Additionally, when $c_1<c<c_2$ equilibrium point $E_2$ is stable, and the value of breast cancer cells at the equilibrium decreases, as we increase the value of $c$, until it reaches 0, when $c>c_2$.

\begin{figure}[h!]
 \makebox[\textwidth][c]{\includegraphics[width=1\textwidth]{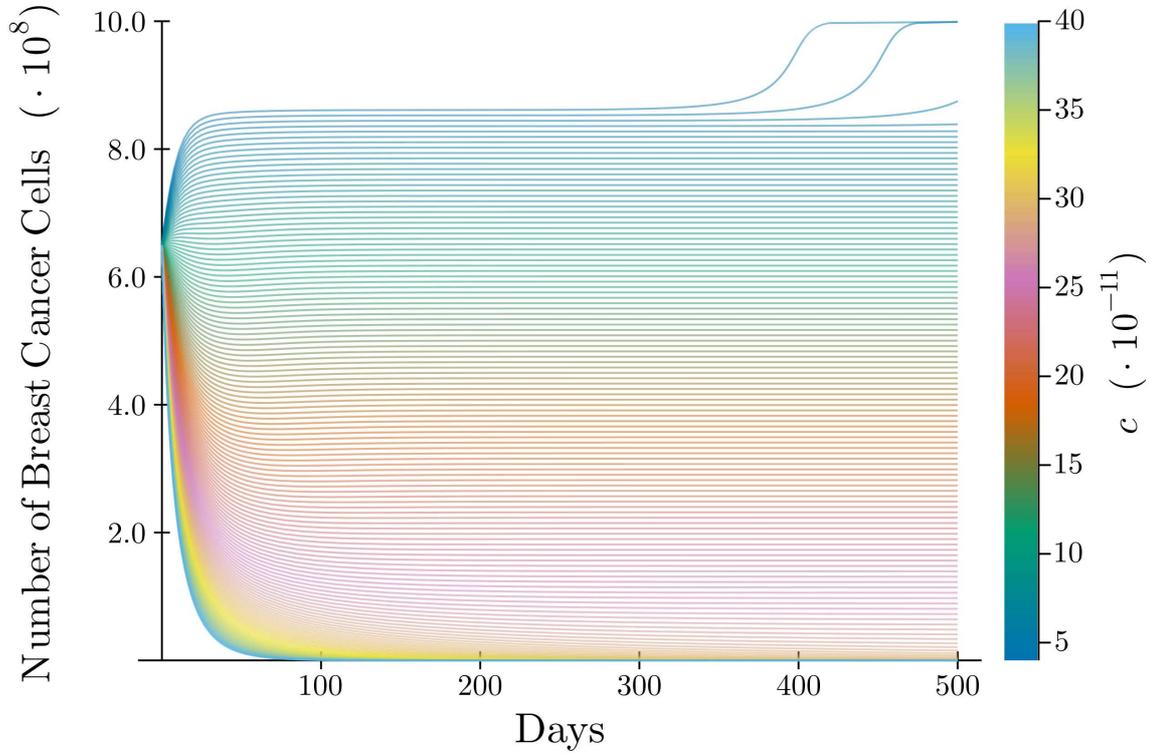}}%
     \caption{Ensemble simulations of IVP, with $c$ taking values in the 150-point discretisation of the interval $\left[ 4 \cdot 10 ^{-11}, 4 \cdot 10^{-10} \right] \; \cdot \; \text{cell}^{-1}  \, \cdot \, \text{day}^{-1}$, $m_T = m_T^B$,   $T_0 = 6.5 \cdot 10^8$ cells, and $B_0 = 50$ cells. The phenomenon of tumour dormancy is evident, since there exist stable equilibrium points where breast-cancer-cell levels are very low.}
  \label{Fig12}
\end{figure}

Additionally, when the value of $c$ is a bit smaller than $c_2$, our model predicts breast cancer cells decreasing to a very small, but nonzero, amount. This clinically undetectable population could seize its proliferation, until the tumour micro-environmental conditions are suitable for its growth. This phenomenon is known as cancer dormancy and it happens to 20–45\% of breast cancer patients \cite{aguirre2007models}. 

To sum up, similarly to \hyperref[ConclusionA]{Conclusion $1_a$} we have that:

\noindent\fbox{\parbox{\textwidth}{\begin{center}\textbf{Conclusion $1_p$:\\} In the presence of tBregs, the breast tumour will reach its carrying capacity due to NK cell insufficiency, i.e. $T \nearrow \frac{1}{b}$ when $c \searrow 0$.  \end{center}}} \label{ConclusionA'}

\subsection{Sensitivity Analysis}
In this section, we perform numerical parameter sensitivity analysis in order to bring light to the dependence of the solution of our model to its parameters. Our approach is to fix all parameter values but one, which we increase and decrease by 10\%, and measure the percent change of the breast cancer cell population after 200 days, when compared to the same value, but for our original parameter. We execute this procedure for four different tBregs initial conditions: $B_0 = 10^6,\; 6 \cdot 10^6, \; 11 \cdot 10^6 \; \text{and} \; 11.5 \cdot 10^6$ cells. Additionally, we have that $T_0 = 8 \cdot 10^8$ cells, $c =3 \cdot 10^{-10}\; \text{cell}^{-1}  \, \cdot \, \text{day}^{-1}$, $m_T = 5 \cdot 10^{-10}\; \text{cell}^{-1}  \, \cdot \, \text{day}^{-1}$.

The results are given in \hyperref[Fig13]{Figure \ref*{Fig13}}. In most of our experiments, the parameter showing the biggest sensitivity is, the rate at which NK cells lyse cancer cells, $c$. The more we increase $c$, the biggest the decrease in final tumour size after 200 days, further supporting \hyperref[ConclusionA]{Conclusion $1_a$} and  \hyperref[ConclusionA']{Conclusion $1_p$}. Another parameter showing high sensitivity is, the tumour growth rate, $a$, which is to be expected. Additionally, the rest of parameters directly involving NK cells, namely, the constant source of NK cells, $\sigma$, the rate of programmable NK cell death, $\theta_N$, and the rate of NK cell death due to Tregs, $\gamma$, all show considerable high sensitivity, which underlines the anti-tumour effect of NK cells. As far as the sensitivity of the parameters regarding Tregs is concerned, it remains almost the same, no matter the initial conditions of tBregs. On the contrary, the parameters regarding tBregs, namely, the rate of tBreg-induced Treg activation, $m_B$, the rate of programmable tBreg cell death, $\theta_B$, and the rate of breast-cancer-induced tBreg activation, $m_T$, all show increasing sensitivity, as we increase the initial number of tBregs. This comes as a further confirmation of \hyperref[conclusionB]{Conclusion 2} and \hyperref[conclusionC]{Conclusion 3}, since the higher the number of tBregs in the body, the bigger the resulting tumour, and therefore the probability of cancer metastasis is increased.

\begin{figure}[h!]
 \makebox[\textwidth][c]{\includegraphics[width=1\textwidth]{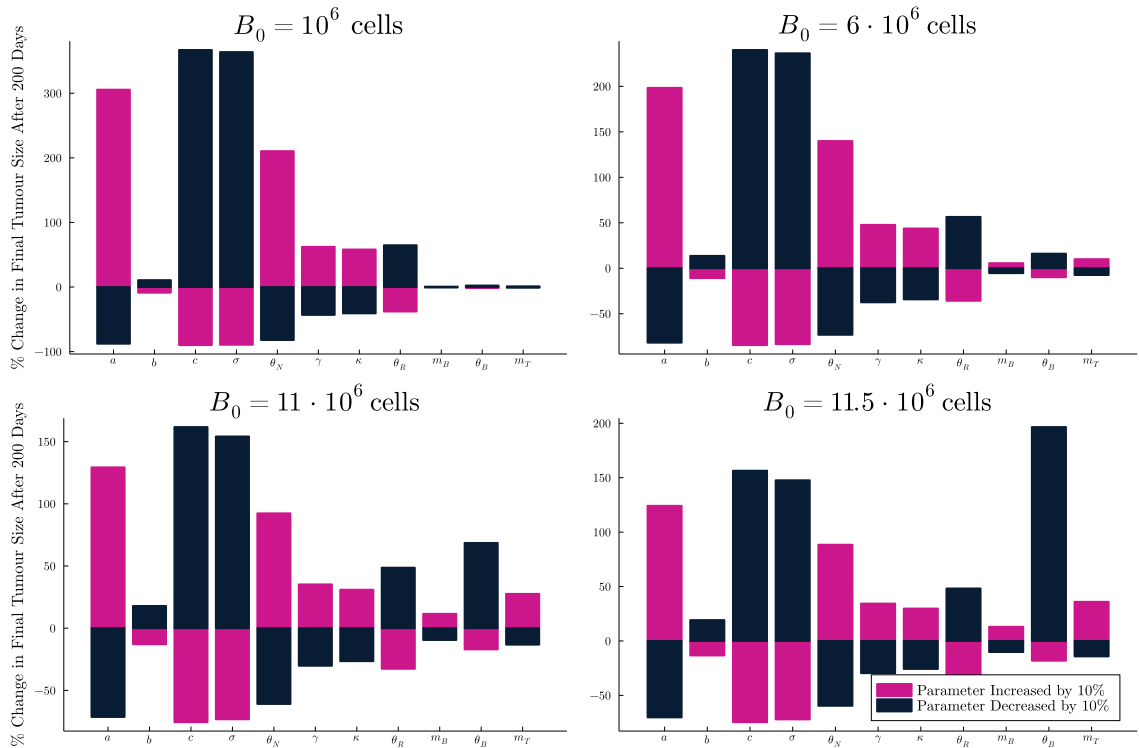}}%
     \caption{Numerical sensitivity of the parameters of IVP, with $T_0 = 8 \cdot 10^8$ cells, $c =3 \cdot 10^{-10}\; \text{cell}^{-1}  \, \cdot \, \text{day}^{-1}$, $m_T = 5 \cdot 10^{-10}\; \text{cell}^{-1}  \, \cdot \, \text{day}^{-1}$. The percentage change in final tumour size after 200 days is plotted. We notice that in all experiments, the parameters regarding NK cells show very high sensitivity, whereas the parameters regarding tBregs show increasing sensitivity as we increase $B_0$.} 
  \label{Fig13}
\end{figure}

\section{Conclusion and Discussion} \label{Discussion}

In this study, we developed a model of non-linear ordinary differential equations, with the goal of capturing the dynamics between breast cancer cells, NK cells, Tregs and the newly discovered tBregs. An introductory approach was taken place, where the functional responses of our model were chosen to be linear (i.e. Holling's type I). 

We showed the existence of three biologically realistic equilibria: an equilibrium with no cancer cells, an equilibrium with cancer cells and no tBregs, and an equilibrium with both cancer cells and tBregs. Using the linearisation and the center manifold theorem, we gave conditions regarding the local stability of each equilibrium.

Using bifurcation analysis, we showed the importance of the rate of NK-induced tumour death, $c$, - independently of the presence or absence of tBregs - on the stability and the existence of the equilibria, which we further proved using our numerical simulations, arriving at \hyperref[ConclusionA]{Conclusion $1_a$} and \hyperref[ConclusionA']{Conclusion $1_p$}.

Moreover, we showed how the sign of $\frac{\theta_B}{m_T} - \frac{1}{b}$ can change the number of equilibria of our model and, therefore, the dynamics of breast cancer growth. On top of that, we proved that when $\frac{\theta_B}{m_T} > \frac{1}{b}$, our model has no closed orbits, meaning that breast cancer will either be cleared by NK cells or stabilise around a constant number.

Additionally, we performed numerical simulations, in which our model was able to capture some interesting behaviours of breast cancer. On the one hand, breast cancer cells decreased for a period of time and after reaching a minimum value, they started to increase. An opposite behaviour was also observed, with cancer cells increasing and reaching a peak, after which  NK cells finally managed to clear the cancer cell population (\hyperref[Fig8]{Figure \ref*{Fig8}} and \hyperref[Fig11]{Figure \ref*{Fig11}}). On the other hand, our model was able to capture the phenomenon of cancer dormancy, due to the existence of a stable equilibrium, in which the number of breast cancer cells is very small (\hyperref[Fig3]{Figure \ref*{Fig3}} and \hyperref[Fig12]{Figure \ref*{Fig12}}). 

Throughout our numerical simulations, the ability of tBregs to cause breast cancer cells to reach their carrying capacity was underlined, thus making us reach \hyperref[conclusionB]{Conclusion 2} and \hyperref[conclusionC]{Conclusion 3}. In fact, our simulations showed that without tBregs, breast cancer cells can only be stabilised around two of the three equilibria admitted by our analysis ($E_1$ and $E_2$). Furthermore, even in the cases where the tumour micro-environmental conditions did not allow tBregs to grow, but only decrease - which happens when $\frac{\theta_B}{m_T} > \frac{1}{b}$ - a high initial number of tBregs caused breast tumour to reach its carrying capacity. Thus, the important takeaway of this study is that when tBregs do not exist in the body, the only way for the breast tumour to reach its carrying capacity is for the rate of NK-induced breast cancer cell death to be tending to 0, whereas if tBregs exist in the body, a sufficiently large number of tBregs can cause the breast tumour to reach its carrying capacity, regardless of the rate of NK-induced breast cancer cell death. Taking all these into account, a potential tBreg-depleting therapy could minimise the ability of breast cancer to grow and metastasise. 
 
Finally, our numerical sensitivity analysis emphasised \hyperref[conclusionB]{Conclusion 2} and \hyperref[conclusionC]{Conclusion 3}, due to our model's high sensitivity to the parameter regarding the NK-induced death of breast cancer cells, as well as the increasing sensitivity to parameters regarding tBregs, as we increased the initial condition of tBregs.

In spite of our model's ability to be used as a framework within which we can study breast cancer growth with respect to NK cells, Tregs and tBregs, there do exist some limitations. Choosing linear functional responses has the advantage of making the model easier to approach, however this makes our results potentially less realistic. Additionally, as tBregs are recently discovered, data regarding their kinetics do not currently exist; this drove us to estimate the value of the corresponding parameters. However, our findings could serve as a catalyst to more research regarding the function and kinetics of tBregs, which in turn could help further refine the proposed model and, therefore, the understanding of the role of tBregs in breast cancer growth and metastasis.

\section*{Acknowledgement}
The authors would like to thank the reviewers for their thoughtful comments and efforts towards improving our manuscript.

\backmatter

\begin{appendices}

\section{Preliminary Results} \label{AppendixSecPropSol}

In this section, we prove that IVP (we remind that this acronym is used for the initial value problem  $\left\{\eqref{reducedModel},\eqref{reducedModel_ICs}\right\}$) has a unique solution, which is non-negative for non-negative initial conditions and for positive parameter values. We also prove that our solution is global, i.e. it does not explode for some finite positive value of time $t$. The above conditions are necessary in order to assure that IVP yields biologically realistic results.

\begin{secproposition}[Uniqueness and non-negativity]
\label{unnneg}
For every $\left(T_0,N_0,R_0,B_0\right)\in \mathbb{R}_{\ge 0}^4$, IVP  has a unique (local) solution $\left(T, N, R, B\right):\left[0,\tau\right)\to\mathbb{R}_{\ge 0}^4$ for some $\tau>0$.

\end{secproposition}
\begin{proof}
It is easy to see that the conditions of the Picard-Lindelöf theorem are fulfilled, since every function of the right-hand side of system \eqref{reducedModel} is continuous, just like its partial derivative with respect to every variable. Thus, we have that there exists a unique solution to IVP. In fact, we can extended the solution and consider it in the maximal non-negative interval of existence.

Next, we prove the non-negativity  of the solution. Rewriting \eqref{dT} in the following form
\begin{equation*}
    {\dv{T}{t}}(t) +(cN(t)-a)T(t)=-abT^2(t) \,, 
\end{equation*}
we notice that we have a Bernoulli equation for the variable $T$, thus its solution is
\begin{equation*}
    T(t) = \frac{T_0 e^{\int_0^t \left( a-cN(s) \right) \, \mathrm{d}s}}{1+T_0ab \int_0^t \! e^{{\int_0^s \left( a-cN(\xi) \right) \, \mathrm{d}\xi}} \,  \mathrm{d}s} \,, %\label{dTSolutionReduced}
\end{equation*}
which is non-negative, if $T_0$ is non-negative.

Using the fact that $\sigma > 0$, we turn our attention to \eqref{dN}. We have that
\begin{equation*}
  {\dv{N}{t}}(t) > -(\gamma R(t) + \theta_N)N(t) \,,    
\end{equation*}
and using Grönwall's inequality, we have that
\begin{equation*}
   N(t) \ge N_0e^{- \int_0^t \left( \gamma R(s) + \theta_N \right) \, \mathrm{d}s} \,,    
\end{equation*}
which means that $N(t)\ge0$, when $N_0\ge0$.

Using a similar method as above, from \eqref{dR} we get
\begin{equation*}
    R(t) \ge R_0e^{\int_0^t \left( m_BB(s)-\theta_R \right) \, \mathrm{d}s} \,,   
\end{equation*}
which means that $R(t)\ge0$, when $R_0\ge0$.

We then use the separation of variables method to solve \eqref{dB} for the variable $B$. Its solution is
\begin{equation}
\label{Bfor}
    B(t) = B_0e^{\int_0^t  \left( m_TT(s) - \theta_B \right) \, \mathrm{d}s} \,.
\end{equation}
Clearly, if $B_0\ge0$, then $B(t)\ge0$.

\end{proof}

\begin{secproposition}[Boundedness of $T$]
\label{BofT}
The set $\left[0,1/b\right]$ is positively invariant for the component $T$ of the solution of IVP. 
\end{secproposition}
\begin{proof}
Since the solution of IVP is non-negative, we have from \eqref{dT} that
\begin{equation*}
    {\dv{T}{t}}(t)  = \underbrace{aT(1-bT) - cNT}_{q_1(T)} \le \underbrace{aT(1-bT)}_{q_2(T)} \; . 
\end{equation*}
We assume the following two initial value problems:
\begin{align}
    & \dv{T}{t}  = q_1(T), \quad  T(0)=T_0 \ge 0 \quad \text{ and} \\
    & \dv{y_2}{t} = q_2(y_2), \quad y_2(0) = \frac{1}{b} \; \label{2ndIVPComparison}.
\end{align}
Assuming that $T_0 \le y_2(0) = 1/b $, and since $q_1, q_2$ are Lipschitz functions on $\mathbb{R}$ that satisfy the inequality $ q_1(T) \le  q_2(T)$, from the comparison theorem we have that $T(t) \le y_2(t)$ for $t$ in the maximal non-negative interval of existence of the solution of IVP. Solving initial value problem \eqref{2ndIVPComparison}, yields $y_2 = 1/b$. Hence, $T(t) \le 1/b$, with the assumption that the initial value of $T$ is smaller or equal to $1/b$.
\end{proof}

\begin{secproposition}[Globality]
If $T_0 \le 1/b$, then the solution of IVP is global.
\end{secproposition}
\begin{proof}
From \eqref{dB} and using the fact that $T(t) \le 1/b$, we have that
\begin{equation*}
    {\dv{B}{t}}(t) \le -\theta_BB(t)+\frac{m_T}{b}B(t) \,,
\end{equation*}
and by using Grönwall's inequality, we have that
\begin{equation*}
    B(t) \le B_0e^{(\frac{m_T}{b}-\theta_B)t} \,.
\end{equation*}

Moreover, from \eqref{dR} and using the fact that $\theta_RR \ge 0 $ we get
\begin{equation*}
    {\dv{R}{t}}(t) \le \kappa + m_BB(t)R(t) \,,
\end{equation*}
and by using Grönwall's inequality, we have that
\begin{equation*}
    R(t) \le e^{m_B \int_0^t \! B(s)  \, \mathrm{d}s} \left(R_0 + \kappa \int_0^t \! e^{ - m_B \int_0^\xi \! B(s)  \, \mathrm{d}s} \, \mathrm{d}\xi  \right) \,.
\end{equation*}

Finally, from \eqref{dN} and using the fact that $ \gamma RN \ge 0 $ we get
\begin{equation*}
    {\dv{N}{t}}(t)\le \sigma - \theta_N N(t) \,,
\end{equation*}
and by using Grönwall's inequality, we have that
\begin{equation*}
    N(t) \le \frac{\sigma}{\theta_N} - \frac{\sigma}{\theta_N}e^{-\theta_Nt} + e^{-\theta_Nt}N_0 \,.
\end{equation*}

Since the solution is bounded on any compact non-negative interval, we deduce its (positive) globality. 
\end{proof}
%%%%%%%%%%%%%%%%%%%%%%%%%%%%%%%%%%%%%%%%%%%%%%}

%%%%%%%%%%%%%%%%%%%%

\section{Parameter Estimation}\label{secB1}
Here we explain our reasoning behind our choice of parameters. Most of the parameters in our model have been chosen based on methods and data that can also be found in \cite{bitsouni2021mathematical}.  

\subsection{The Tumour} \label{sectionParameterEstimationTumourReducedModel}
Based on the data fitting experiments conducted in \cite{bitsouni2021mathematical}, we chose the logistic function to model the breast cancer growth, with the tumour growth rate being $a = 0.15$ day$^{-1}$  and the inverse of the tumour carrying capacity being $b = 1 \cdot 10^{-9}$ cell$^{-1}$. The cell lines used for the estimation of these parameters are CN34BrM, MDA-231 and SUM1315.

\subsection{The NK Cells} \label{reducedModelNKParameterEstimation}
Healthy young adults have a total NK production rate of (15 ± 7.6)·10$^6$ cell · litre$^{-1}$ · day$^{-1}$, while healthy older adults have one of (7.3 ± 3.7)·10$^6$ cells · litre$^{-1}$ · day$^{-1}$ \cite{zhang}. Since the average amount of blood in the human body is about 5 litre \cite{starr2012biology}, the constant source of NK cells is in the range 
\begin{equation*}
    \sigma \in \left[  1.8 \cdot10^7 \text{ cell} \cdot \text{day}^{-1},\,  1.13 \cdot10^8 \text{ cell} \cdot \text{day}^{-1} \right]\,.
\end{equation*}

The half-life of NK cells in humans is 1 to 2 weeks \cite{zhang} which, assuming exponential decay of NK cells, yields a range for $\theta_N$ of $({\frac{\ln2}{14}}, {\frac{\ln2}{7}})=(0.049, 0.099)$. Here, we choose an NK cells half-life of 11 days with a corresponding programmable NK death rate of
\begin{equation*}
    \theta_N = \frac{\ln2}{11 \text{ day}} \approx 6.301 \cdot 10^{-2} \text{ day}^{-1}\,.
\end{equation*}

Approximately 4 to 29\% of circulating lymphocytes are NK cells \cite{hema}. The average number of lymphocytes per microlitre is 1000 to 4800 cells \cite{abbas2014cellular}, and since the average human has an average of 5 litres of blood, we have that the total population of lymphocytes in a human is 5·10$^9$ to 24·10$^9$ cells. Therefore, the total population of NK cells in blood is $N_{min}=$ 2·10$^8$ to $N_{max}=$ 6.96·10$^9$ cells. At the healthy equilibrium, our model suggests the population of NK cells to be  $\frac{\sigma  \theta _R}{\gamma  \kappa +\theta _N \theta _R}$, which means that
\begin{equation*}
    N_{min} \le \frac{\sigma  \theta _R}{\gamma  \kappa +\theta _N \theta _R} \le N_{max} \Leftrightarrow   \frac{\theta_R ( \sigma - \theta_N )}{N_{max} \kappa} \le \gamma \le \frac{\theta_R ( \sigma - \theta_N )}{N_{min} \kappa}\,.
\end{equation*}
Replacing the minimum value of $\sigma$, and the maximum value of $\theta_N$ and $\kappa$ in the above inequality yields the minimum value of $\gamma$, while replacing the maximum value of $\sigma$, and the minimum value of $\theta_N$ and $\kappa$ (we derive a range for $\kappa$ in \hyperref[reducedModelTregParameterEstimation]{Appendix \ref*{reducedModelTregParameterEstimation}}), yields the maximum value of $\gamma$. Finally, after calculating the above two quantities, the resulting range for the parameter $\gamma$ is
\begin{equation*}
   \gamma \in \left[   1.796 \cdot 10^{-12} \text{ cell}^{-1} \cdot \text{day}^{-1} , 4.52 \cdot 10^{-9} \text{ cell}^{-1} \cdot \text{day}^{-1} \right]\,.
\end{equation*}

\subsection{The Tregs} \label{reducedModelTregParameterEstimation}
For the constant source of Tregs, our model suggests that a healthy organism has an average number of $ \kappa/\theta_R $ Tregs, since this is the coordinate which corresponds to Tregs in the healthy equilibrium. From \cite{pang2013frequency} we get that Tregs are 5 to 10\% of the total CD4$^+$ T cells population circulating in blood, while from \cite{abbas2014cellular} we get that the percentage of CD4$^+$ T cells among the total population of circulating lymphocytes ranges from 50 to 60\%. Hence, the percentage of Tregs among the total population of circulating lymphocytes is 2.5 to 6\%. The average number of lymphocytes per microlitre is 1000 to 4800 cells \cite{abbas2014cellular}, and since the average human has an average of 5 litres of blood, we have that the total population of lymphocytes in a human is 5·10$^9$ to 24·10$^9$ cells. Therefore, the total population of Tregs in blood is 1.25·10$^8$ to 1.44·10$^9$ cells. Solving the equation
\begin{equation}
    \textit{total population of Tregs} = \frac{\kappa}{\theta_R}, \notag
\end{equation}
for $\kappa$ and replacing the range of values for the total population of Tregs as found above and the value of $\theta_R$ as found in the following paragraph, we finally get the constant source of Tregs to be in the interval
\begin{equation}
    \kappa \in [4.8137 \cdot 10^6 \text{ cell·day$^{-1}$},5.5454 \cdot 10^7 \text{ cell·day$^{-1}$}]\,. \notag
\end{equation}

The half-life of Tregs is found to be about 18 days \cite{mabarrack2008}. Thus, assuming Tregs follow exponential decay we have that
\begin{equation}
    \theta_R = \frac{\ln2}{18 \text{ day}} \approx 3.851 \cdot 10^{-2} \text{ day}^{-1}. \notag
\end{equation}

The rate of NK cell death due to Tregs, $\gamma$, is assumed to be
\begin{equation*}
    \gamma = 1 \cdot 10^{-10} \text{ cell}^{-1} \cdot \text{day}^{-1}.
\end{equation*}

\subsection{The tBregs}
As B cells are less studied than T cells and NK cells, we are not able to find the half-life of Bregs, let alone tBregs, as they are recently discovered. So, we estimate the rate of programmable tBreg cell death to be around $\theta_B = 0.4 \text{ day}^{-1}$.

The same holds for the rate of breast-cancer-induced tBreg activation, $m_T$, which we assume it to be either  $m_T = 5.2 \cdot 10^{-15}$ cell$^{-1}$ $\cdot$ day$^{-1}$ or  $m_T = 5 \cdot 10^{-10}$ cell$^{-1}$ $\cdot$ day$^{-1}$.

\end{appendices}
%%%%%%%%%%%%%%%%%%%%%%%%%%%%%%%%%%%%%%%%%%%%%%%
\bibliographystyle{abbrv}
\bibliography{mybibfile}\label{bibliography}

\begin{thebibliography}{10}

\bibitem{abbas2014cellular}
A.~K. Abbas, A.~H. Lichtman, and S.~Pillai.
\newblock {\em Cellular and Molecular Immunology E-book}.
\newblock Elsevier Health Sciences, 2014.

\bibitem{aguirre2007models}
J.~A. Aguirre-Ghiso.
\newblock Models, mechanisms and clinical evidence for cancer dormancy.
\newblock {\em Nature Reviews Cancer}, 7(11):834--846, 2007.

\bibitem{al2020modeling}
S.~M. Al-Tuwairqi, N.~O. Al-Johani, and E.~A. Simbawa.
\newblock Modeling dynamics of cancer virotherapy with immune response.
\newblock {\em Advances in Difference Equations}, 2020(1):1--26, 2020.

\bibitem{biragyn2014generation}
A.~Biragyn, C.~Lee-Chang, and M.~Bodogai.
\newblock Generation and identification of tumor-evoked regulatory {B} cells.
\newblock {\em Methods in Molecular Biology (Clifton, NJ)}, 1190:271--289,
  2014.

\bibitem{bitsouni2021mathematical}
V.~Bitsouni and V.~Tsilidis.
\newblock Mathematical modeling of tumor-immune system interactions: the effect
  of rituximab on breast cancer immune response.
\newblock {\em Journal of Theoretical Biology}, 2022.

\bibitem{bray2018global}
F.~Bray, J.~Ferlay, I.~Soerjomataram, R.~L. Siegel, L.~A. Torre, and A.~Jemal.
\newblock Global cancer statistics 2018: {GLOBOCAN} estimates of incidence and
  mortality worldwide for 36 cancers in 185 countries.
\newblock {\em {CA}: {A} {C}ancer {J}ournal for {C}linicians}, 68(6):394--424,
  2018.

\bibitem{bunimovich2007}
S.~Bunimovich-Mendrazitsky, E.~Shochat, and L.~Stone.
\newblock Mathematical model of {BCG} immunotherapy in superficial bladder
  cancer.
\newblock {\em Bulletin of Mathematical Biology}, 69(6):1847--1870, 2007.

\bibitem{byrne2004macrophage}
H.~M. Byrne, S.~M. Cox, and C.~Kelly.
\newblock Macrophage-tumour interactions: in vivo dynamics.
\newblock {\em Discrete \& Continuous Dynamical Systems - B}, 4(1):81, 2004.

\bibitem{carol2014breast}
D.~Carol, S.~Rebecca, B.~Priti, and J.~Ahmedin.
\newblock Breast cancer statistics, 2013.
\newblock {\em CA: A Cancer Journal for Clinicians}, 64(1):52--62, 2014.

\bibitem{chraa}
D.~Chraa, A.~Naim, D.~Olive, and A.~Badou.
\newblock {T} lymphocyte subsets in cancer immunity: friends or foes.
\newblock {\em Journal of Leukocyte Biology}, 105(2):243--255, 2019.

\bibitem{dePillis2013}
L.~G. de~Pillis, T.~Caldwell, E.~Sarapata, and H.~Williams.
\newblock Mathematical modeling of regulatory {T} cell effects on renal cell
  carcinoma treatment.
\newblock {\em Discrete \& Continuous Dynamical Systems - B}, 18:915--943, 06
  2013.

\bibitem{depillis2003mathematical}
L.~G. de~Pillis and A.~E. Radunskaya.
\newblock A mathematical model of immune response to tumor invasion.
\newblock In {\em Computational Fluid and Solid Mechanics 2003}, pages
  1661--1668. Elsevier, 2003.

\bibitem{dePillis2009}
L.~G. de~Pillis, K.~Renee~Fister, W.~Gu, C.~Collins, M.~Daub, D.~Gross,
  J.~Moore, and B.~Preskill.
\newblock Mathematical model creation for cancer chemo-immunotherapy.
\newblock {\em Computational and Mathematical Methods in Medicine},
  10(3):165--184, 2009.

\bibitem{denbreems2016}
N.~Y. {den Breems} and R.~Eftimie.
\newblock The re-polarisation of {M2} and {M1} macrophages and its role on
  cancer outcomes.
\newblock {\em Journal of Theoretical Biology}, 390:23--39, 2016.

\bibitem{eftimie2011interactions}
R.~Eftimie, J.~L. Bramson, and D.~J. Earn.
\newblock Interactions between the immune system and cancer: a brief review of
  non-spatial mathematical models.
\newblock {\em Bulletin of Mathematical Biology}, 73(1):2--32, 2011.

\bibitem{fentiman2006male}
I.~S. Fentiman, A.~Fourquet, and G.~N. Hortobagyi.
\newblock Male breast cancer.
\newblock {\em The Lancet}, 367(9510):595--604, 2006.

\bibitem{ferlay2019estimating}
J.~Ferlay, M.~Colombet, I.~Soerjomataram, C.~Mathers, D.~Parkin,
  M.~Pi{\~n}eros, A.~Znaor, and F.~Bray.
\newblock Estimating the global cancer incidence and mortality in 2018:
  {GLOBOCAN} sources and methods.
\newblock {\em {I}nternational {J}ournal of {C}ancer}, 144(8):1941--1953, 2019.

\bibitem{ghosh2018mathematical}
S.~Ghosh and S.~Banerjee.
\newblock Mathematical modeling of cancer--immune system, considering the role
  of antibodies.
\newblock {\em Theory in Biosciences}, 137(1):67--78, 2018.

\bibitem{gonzalez2018roles}
H.~Gonzalez, C.~Hagerling, and Z.~Werb.
\newblock Roles of the immune system in cancer: from tumor initiation to
  metastatic progression.
\newblock {\em Genes \& Development}, 32(19-20):1267--1284, 2018.

\bibitem{guo}
F.~F. Guo and J.~W. Cui.
\newblock The role of tumor-infiltrating {B} cells in tumor immunity.
\newblock {\em Journal of Oncology}, 2019:1--9, 2019.

\bibitem{he2017mathematical}
D.-H. He and J.-X. Xu.
\newblock A mathematical model of pancreatic cancer with two kinds of
  treatments.
\newblock {\em Journal of Biological Systems}, 25(01):83--104, 2017.

\bibitem{jordan2007nonlinear}
D.~Jordan, P.~Smith, and P.~Smith.
\newblock {\em Nonlinear Ordinary Differential Equations: An Introduction for
  Scientists and Engineers}.
\newblock Oxford University Press, 2007.

\bibitem{hema}
E.~Keohane, L.~Smith, and J.~Walenga.
\newblock {\em Rodak's Hematology - E-Book: Clinical Principles and
  Applications}.
\newblock Elsevier Health Sciences, 2015.

\bibitem{khailaie2013mathematical}
S.~Khailaie, F.~Bahrami, M.~Janahmadi, P.~Milanez-Almeida, J.~Huehn, and
  M.~Meyer-Hermann.
\newblock A mathematical model of immune activation with a unified self-nonself
  concept.
\newblock {\em Frontiers in Immunology}, 4:474, 2013.

\bibitem{kuznetsov1994}
V.~A. Kuznetsov, I.~A. Makalkin, M.~A. Taylor, and A.~S. Perelson.
\newblock Nonlinear dynamics of immunogenic tumors: parameter estimation and
  global bifurcation analysis.
\newblock {\em Bulletin of Mathematical Biology}, 56(2):295--–321, Mar 1994.

\bibitem{lopez2014validated}
A.~G. L{\'o}pez, J.~M. Seoane, and M.~A. Sanju{\'a}n.
\newblock A validated mathematical model of tumor growth including tumor--host
  interaction, cell-mediated immune response and chemotherapy.
\newblock {\em Bulletin of Mathematical Biology}, 76(11):2884--2906, 2014.

\bibitem{mabarrack2008}
N.~H.~E. Mabarrack, N.~L. Turner, and G.~Mayrhofer.
\newblock Recent thymic origin, differentiation, and turnover of regulatory {T}
  cells.
\newblock {\em Journal of Leukocyte Biology}, 84(5):1287--1297, Nov 2008.

\bibitem{makhlouf2020mathematical}
A.~M. Makhlouf, L.~El-Shennawy, and H.~A. Elkaranshawy.
\newblock Mathematical modelling for the role of {CD4+} {T} cells in
  tumor-immune interactions.
\newblock {\em Computational and Mathematical Methods in Medicine}, 2020, 2020.

\bibitem{Nani_Freedman_2000}
F.~Nani and H.~Freedman.
\newblock A mathematical model of cancer treatment by immunotherapy.
\newblock {\em Mathematical Biosciences}, 163(2):159--199, Feb 2000.

\bibitem{narendra2013immune}
B.~L. Narendra, K.~E. Reddy, S.~Shantikumar, and S.~Ramakrishna.
\newblock Immune system: a double-edged sword in cancer.
\newblock {\em Inflammation Research}, 62(9):823--834, 2013.

\bibitem{olkhanud2009breast}
P.~B. Olkhanud, D.~Baatar, M.~Bodogai, F.~Hakim, R.~Gress, R.~L. Anderson,
  J.~Deng, M.~Xu, S.~Briest, and A.~Biragyn.
\newblock Breast cancer lung metastasis requires expression of chemokine
  receptor {CCR4} and regulatory {T} cells.
\newblock {\em Cancer Research}, 69(14):5996--6004, 2009.

\bibitem{olkhanud2011tumor}
P.~B. Olkhanud, B.~Damdinsuren, M.~Bodogai, R.~E. Gress, R.~Sen, K.~Wejksza,
  E.~Malchinkhuu, R.~P. Wersto, and A.~Biragyn.
\newblock Tumor-evoked regulatory {B} cells promote breast cancer metastasis by
  converting resting {CD4+} {T} cells to {T}-regulatory cells.
\newblock {\em Cancer Research}, 71(10):3505--3515, 2011.

\bibitem{pang2013frequency}
H.~Pang, Q.~Yu, B.~Guo, Y.~Jiang, L.~Wan, J.~Li, Y.~Wu, and K.~Wan.
\newblock Frequency of regulatory {T}-cells in the peripheral blood of patients
  with pulmonary tuberculosis from {Shanxi} province, {China}.
\newblock {\em PLOS ONE}, 8(6):e65496, 2013.

\bibitem{perko2006differential}
L.~Perko.
\newblock {\em Differential Equations and Dynamical Systems, Third Edition}.
\newblock Springer Science \& Business Media, 3rd edition, 2006.

\bibitem{phan2017role}
T.~A. Phan and J.~P. Tian.
\newblock The role of the innate immune system in oncolytic virotherapy.
\newblock {\em Computational and Mathematical Methods in Medicine}, 2017, 2017.

\bibitem{rackauckas2017differentialequations}
C.~Rackauckas and Q.~Nie.
\newblock Differentialequations.jl -- a performant and feature-rich ecosystem
  for solving differential equations in {Julia}.
\newblock {\em Journal of Open Research Software}, 5(1):15, 2017.

\bibitem{rosser2015regulatory}
E.~C. Rosser and C.~Mauri.
\newblock Regulatory {B} cells: origin, phenotype, and function.
\newblock {\em Immunity}, 42(4):607--612, 2015.

\bibitem{segel}
L.~Segel and I.~Cohen.
\newblock {\em Design Principles for the Immune System and Other Distributed
  Autonomous Systems}.
\newblock Santa Fe Institute Studies on the Sciences of Complexity. Oxford
  University Press, 2001.

\bibitem{senekal2021natural}
N.~S. Senekal, K.~J. Mahasa, A.~Eladdadi, L.~de~Pillis, and R.~Ouifki.
\newblock Natural killer cells recruitment in oncolytic virotherapy: a
  mathematical model.
\newblock {\em Bulletin of Mathematical Biology}, 83(7):1--51, 2021.

\bibitem{sopik2018relationship}
V.~Sopik and S.~A. Narod.
\newblock The relationship between tumour size, nodal status and distant
  metastases: on the origins of breast cancer.
\newblock {\em Breast Cancer Research and Treatment}, 170(3):647--656, 2018.

\bibitem{starr2012biology}
C.~Starr, R.~Taggart, and C.~Evers.
\newblock {\em Biology: The Unity and Diversity of Life}.
\newblock Cengage Learning, 2012.

\bibitem{szymanska2003}
Z.~Szyma{\'n}ska.
\newblock Analysis of immunotherapy models in the context of cancer dynamics.
\newblock {\em International Journal of Applied Mathematics and Computer
  Science}, 3(13):407--418, 2003.

\bibitem{vaghi2020carrying}
C.~Vaghi, A.~Rodallec, R.~Fanciullino, J.~Ciccolini, J.~P. Mochel, M.~Mastri,
  C.~Poignard, J.~M.~L. Ebos, and S.~Benzekry.
\newblock Population modeling of tumor growth curves and the reduced gompertz
  model improve prediction of the age of experimental tumors.
\newblock {\em PLOS Computational Biology}, 16(2):1--24, 02 2020.

\bibitem{villasana2003delay}
M.~Villasana and A.~Radunskaya.
\newblock A delay differential equation model for tumor growth.
\newblock {\em Journal of Mathematical Biology}, 47(3):270--294, 2003.

\bibitem{zhang}
Y.~Zhang, D.~L. Wallace, C.~M. De~Lara, H.~Ghattas, B.~Asquith, A.~Worth, G.~E.
  Griffin, G.~P. Taylor, D.~F. Tough, P.~C. Beverley, et~al.
\newblock In vivo kinetics of human natural killer cells: the effects of ageing
  and acute and chronic viral infection.
\newblock {\em Immunology}, 121(2):258--265, 2007.

\end{thebibliography}
\addcontentsline{toc}{chapter}{Bibliography}
%%%%%%%%%%%%%%%%%%%%%%%%%%%%%%%%%%%%%%%%%%%%%%%

\end{document}